\documentclass[letterpaper,11pt]{article}
\usepackage[margin=1in]{geometry}

\usepackage[utf8]{inputenc}
\usepackage{tikz}
\usepackage{bm}
\usepackage[appendix=inline]{apxproof}
\usepackage{relsize}
\usepackage{stackengine}
\stackMath
\usepackage{amsmath}
\usepackage{mathtools}
\usepackage{amssymb}
\usepackage{todonotes}
\usepackage{hyperref}
\hypersetup{
	    colorlinks,
	    linkcolor={red!50!black},
	    citecolor={blue!50!black},
	    urlcolor={blue!30!black},
	    breaklinks=true
	    }
\usepackage{mathabx}
\usepackage{stackengine}
\stackMath
\usepackage{lineno}
\usepackage[strict]{changepage}
\usepackage{nicefrac}
\usepackage{natbib}
\usepackage{xspace}
\usepackage{multirow}
\usepackage{colonequals}
\usepackage{booktabs}
\usepackage{bm}
\usepackage{tikz}
\usetikzlibrary{arrows,positioning}
\usepackage{subcaption}
\usepackage{graphicx}
\usepackage{makecell}

\newtheoremrep{theorem}{Theorem}[section]

\newtheoremrep{proposition}[theorem]{Proposition}
\newtheoremrep{lemma}[theorem]{Lemma}
\newtheoremrep{example}[theorem]{Example}
\newtheoremrep{definition}[theorem]{Definition}

\newcommand{\NN}{\mathbb{N}}

\newcommand\vars{\mathsf{Vars}}

\newcommand\sat{\mathsf{Models}}
\newcommand\ssat{\#\mathsf{Models}}

\newcommand\assign{\mathsf{Assign}}

\newcommand{\rvars}[1]{\ensuremath{\bm{#1}}\xspace}
\newcommand{\X}{\rvars{X}}
\newcommand{\Y}{\rvars{Y}}
\newcommand{\Z}{\rvars{Z}}

\usetikzlibrary{shapes}
\tikzset{
	rect/.style={
		rectangle,
		rounded corners,
		draw=black, 
		thick,
		text centered},
	rectw/.style={
		rectangle,
		rounded corners,
		draw=white, 
		thick,
		text centered},
	sq/.style={
		rectangle,
		draw=black, 
		thick,
		text centered},
	sqw/.style={
		rectangle,
		thick,
		text centered},
	arrout/.style={
		->,
		-latex,
		thick,
	},
	arrin/.style={
		<-,
		latex-,
		thick,
		El         },
	arrd/.style={
		<->,
		>=latex,
		thick,
	},
	arrw/.style={
		thick,
	}
}

\tikzset{
	circ/.style={
		circle,
		draw=black, 
		thick,
		text centered,
	},
	circw/.style={
		circle,
		draw=white, 
		thick,
		text centered,
	},
	arrout/.style={
		->,
		-latex,
		thick,
	},
	arrin/.style={
		<-,
		latex-,
		thick,
	},
	arrw/.style={
		-,
		thick,
	},
	arrww/.style={
		-,
		thick,
		draw=white, 
	}
}

\newcommand{\as}{\boldsymbol{a}}
\newcommand{\ba}{\mathbf{a}}

\newcommand{\w}{\mathbf{w}}
\newcommand{\A}{A}

\newcommand{\feat}[1]{\mbox{\bf \footnotesize #1}}

\begin{document}

\title{A Circus of Circuits:\texorpdfstring{\\}{ }Connections Between Decision Diagrams, Circuits, and Automata}

\author{Antoine Amarilli \and Marcelo Arenas \and YooJung Choi \and Mikaël Monet \and Guy Van den Broeck \and Benjie Wang}

\date{}
\maketitle

\begin{abstract}
  This document is an introduction to two related formalisms to define Boolean
  functions: binary decision diagrams, and Boolean circuits. It presents these
  formalisms and several of their variants studied in the setting of knowledge
  compilation. Last, it explains how these formalisms can be connected to the
  notions of automata over words and trees.
\end{abstract}

\tableofcontents

\pagebreak

\section{Introduction} \label{sec:intro}
This document is about \emph{Boolean functions} and formalisms to represent
them. Given their naturalness, many communities in theoretical and applied
computer science use Boolean functions and study their representations. The
\emph{circuit complexity} community~\citep{vollmer1999introduction}, for
instance, studies classes of \emph{Boolean circuits} defined by conditions such
as the \emph{depth}, the \emph{size}, and the \emph{kinds of gates} that are
allowed. Other communities study \emph{Boolean
functions}~\citep{wegener1987complexity}, or ways to represent them concisely,
in particular as \emph{decision diagrams} \citep{W04}. Meanwhile, the
\emph{knowledge compilation} community, motivated by the practical use case of
\emph{SAT solvers}, investigates formalisms to represent Boolean functions as
\emph{binary decision diagrams}, or \emph{Boolean circuits}, while ensuring the tractability of certain tasks. A central work in this area, but dating back to 2002, is the \emph{knowledge compilation map}~\citep{darwiche2002knowledge}; the area also uses tools from neighboring fields such as \emph{communication complexity}~\citep{kushilevitz1997communication}.
Last, the \emph{database theory} community has also investigated Boolean functions, in particular in the setting of \emph{provenance}~\citep{green2007provenance}, including representations such as \emph{provenance circuits}~\citep{deutch2014circuits}: in this area, restricted Boolean circuit classes were in particular studied for query evaluation on \emph{probabilistic databases}~\citep{jha2013knowledge}.
All told, this illustrates that the notions of \emph{decision diagrams} (e.g., OBDDs) and \emph{restricted circuit classes} (e.g., d-DNNFs) are studied to a large extent by separate communities.

The goal of this document is twofold. First, we propose a unified introduction
to \emph{binary decision diagrams} and \emph{Boolean circuits}, using consistent terminology across the two paradigms. We hope that this can serve as an introduction to researchers familiar with one of the two formalisms, as a way to understand the connections with the other formalism. Second, we present a lesser-known correspondence that relates \emph{automata on words and trees} to these binary decision diagrams and Boolean circuit classes, specifically, to the setting of \emph{ordered} binary decision diagrams and \emph{structured} Boolean circuit classes. The connection goes through the notion of \emph{provenance circuits}, which relates automata to a Boolean function informally describing their behavior on words of a specific length: intuitively, provenance circuits are obtained by unraveling the automaton up to that length. The point of this correspondence is that well-known conditions on finite automata (e.g., determinism, unambiguity) relate to conditions on the resulting ordered binary decision diagrams (for word automata) or structured Boolean circuits (for tree automata).

Thus, our hope is that this document can serve as a ``Rosetta stone'' to highlight connections between binary decision diagrams, Boolean circuits, and automata, and encourage further interaction between the communities studying these formalisms.

\paragraph*{Document structure.} The document is structured in the following way. First, in Section~\ref{sec:prelims}, we give some common preliminaries. We then define the two main formalisms we use to represent Boolean functions: we define \emph{binary decision diagrams} in Section~\ref{sec:diagrams}, and then define \emph{Boolean circuits} in Section~\ref{sec:circuits}. We then explain in Section~\ref{sec:diag_circ} in which sense binary decision diagrams can be seen as a special case of Boolean circuits. Then, we introduce in Section~\ref{sec:automata} the notions of automata on words and on trees, and we explain how we can define \emph{provenance circuits} for automata: this allows us to draw a correspondence between conditions on automata and conditions on provenance circuits.

\paragraph*{Acknowledgements.} This work was initiated while the authors were visiting the Simons Institute for the Theory of Computing, and was done in part during the corresponding program at the institute.

\section{Preliminaries} \label{sec:prelims}
We write sets of variables with uppercase boldface letters, such as $\X,\Y,\Z$, and single
variables with uppercase non-bold letters, such as~$X\in \X, Y\in \Y$, etc.

\paragraph*{Assignments and Boolean functions.}
Let~$\X$ be a finite set of variables.  An \emph{assignment of~$\X$} is a
function~$\ba: \X \to \{0,1\}$. We denote by~$\assign(\X)$ the set of all assignments
of~$\X$.  A \emph{Boolean function over~$\mathbf{X}$}
is a function~$f\colon \assign(\X)\to \{0,1\}$.
An assignment~$\ba$ of $\mathbf{X}$ is \emph{satisfying} if~$f(\ba) =1$ (also
denoted~$\ba \models f$).  We denote by~$\sat(f)\subseteq \assign(\X)$ the set of all
satisfying assignments of~$f$, and~$\ssat(f)$ the size of this set. 

\paragraph*{CNFs and DNFs.}
Let $\X$ a set of variables. A \emph{literal} is an expression of the form $X$
or $\lnot X$ for $X\in \X$. A \emph{clause} is a disjunction of literals, for
instance $X_1 \lor X_2 \lor \lnot X_3$. A \emph{formula in conjunctive normal
form}, or \emph{CNF} for short, is a formula that is a conjunction of clauses,
i.e., an expression of the form $\bigwedge_{i=1}^n C_i$ where each $C_i$ is a
clause. We call a conjunction of literals a \emph{term}. A \emph{formula in
disjunctive normal form}, or DNF, is a disjunction of terms i.e., an expression
of the form $\bigvee_{i=1}^n t_i$ where each $t_i$ is a term.

\paragraph*{Directed graphs, DAGs, labeled DAGs.}
A {\em directed graph} is a tuple $G = (N, E)$ where $N$ is the set of nodes and $E \subseteq N \times N$ is the set of edges. We say that $G$ is \emph{acyclic} if it contains no cycles, and call it a \emph{directed acyclic graph} or \emph{DAG}.

For a finite set $\Sigma$ of labels, a \emph{$\Sigma$-node-labeled directed graph}, or just node-labeled directed graph, is a tuple $G = (N, E, \lambda)$ where $(N, E)$ is a directed graph and $\lambda\colon N\to\Sigma$ is the node labeling function. A \emph{$\Sigma$-labeled directed graph}, or just labeled directed graph, is a node-labeled directed graph whose edges additionally carry labels. Formally, a $\Sigma$-labeled directed graph is a tuple $G = (N, E, \lambda)$ such that $E \subseteq N \times \Sigma \times N$ and $(N, \{(x,y) \mid \exists \ell \colon (x, \ell, y) \in E\}, \lambda)$ is a $\Sigma$-node-labeled directed graph. 
A node $u$ of $G$ is called a {\em source}
if it does not have any incoming edges, a {\em sink} if it does not have any outgoing edges, and an {\em internal node} if it has some outgoing edges.
If a labeled directed graph is acyclic, we call $G$ a 
\emph{labeled DAG}.
Notice that the previous definition allows a labeled directed graph~$G$ to have multiple edges between the same pair of nodes, but each one with a different label.
We define the \emph{size of $G$}, written $|G|$, to be $|N| + |E| + |\Sigma|$.

\section{Binary Decision Diagrams} \label{sec:diagrams}
\newcommand{\true}{\mathsf{t}}
\newcommand{\false}{\mathsf{f}}
\newcommand{\strue}{{\footnotesize $\mathsf{t}$}}
\newcommand{\sfalse}{{\footnotesize $\mathsf{f}$}}

\newcommand{\D}{\mathcal{D}}

\newcommand{\nbdd}{\textsf{nBDD}\xspace}
\newcommand{\nbdds}{\textsf{nBDDs}\xspace}
\newcommand{\ubdd}{\textsf{uBDD}\xspace}
\newcommand{\bdd}{\textsf{BDD}\xspace}
\newcommand{\bdds}{\textsf{BDDs}\xspace}

\newcommand{\nfbdd}{\textsf{nFBDD}\xspace}
\newcommand{\nfbdds}{\textsf{nFBDDs}\xspace}
\newcommand{\Nfbdd}{\textsf{NFBDD}\xspace}
\newcommand{\nobdd}{\textsf{nOBDD}\xspace}

\newcommand{\pbdds}{\textsf{PBDDs}\xspace}

\newcommand{\ndt}{\textsf{nDT}\xspace}
\newcommand{\ndts}{\textsf{nDTs}\xspace}

\newcommand{\ndf}{\textsf{nDF}\xspace}
\newcommand{\ndfs}{\textsf{nDFs}\xspace}

\newcommand{\ufbdd}{\textsf{uFBDD}\xspace}
\newcommand{\ufbdds}{\textsf{uFBDDs}\xspace}
\newcommand{\uobdd}{\textsf{uOBDD}\xspace}

\newcommand{\fbdd}{\textsf{FBDD}\xspace}
\newcommand{\fbdds}{\textsf{FBDDs}\xspace}
\newcommand{\obdd}{\textsf{OBDD}\xspace}
\newcommand{\obdds}{\textsf{OBDDs}\xspace}
\newcommand{\nobdds}{\textsf{nOBDDs}\xspace}

\newcommand{\udt}{\textsf{uDT}\xspace}
\newcommand{\ddt}{\textsf{DT}\xspace}
\newcommand{\ddts}{\textsf{DTs}\xspace}

\newcommand{\nodt}{\textsf{nODT}\xspace}
\newcommand{\uodt}{\textsf{uODT}\xspace}
\newcommand{\odt}{\textsf{ODT}\xspace}

We start by giving formal definitions of (nondeterministic) binary decision diagrams
and their semantics, and present two conditions on diagrams: variable
structuredness (corresponding to free or ordered binary decision diagrams), and
ambiguity levels. We then comment on an alternative way to express
nondeterminism. After that, we present the case of binary decision diagrams
without sharing (aka decision trees), and discuss the notion of completeness for
binary decision diagrams, along with an alternative semantics called the zero-suppressed semantics.

\subsection{Basic classes}
A \emph{nondeterministic binary decision diagram} (\nbdd~\citep{BW97,ACMS20}) over  a
set of variables $\X$ is a 
labeled DAG~$\D$ 
such that: (i) each
edge is labeled with either the symbol $0$ 
(a \emph{$0$-edge})
or the symbol $1$
(a \emph{$1$-edge});
(ii) each sink is labeled
with~$\true$ (for true) or $\false$ (for false); and (iii) each internal
node is labeled
with a variable $X \in \X$ and has at least one outgoing $0$-edge and at least one outgoing $1$-edge. 
Given an assignment $\as \in \assign(\X)$, a \emph{run} $\pi$ of~$\D$ \emph{following} $\as$ is a sequence 
$u_1, u_2, \ldots, u_k$ of nodes of $\D$ such that $u_1$ is a source of $\D$, $u_k$ is a sink of~$\D$, and for each $i \in \{1, \ldots, k - 1\}$, letting $X_i$ be the variable that labels $u_i$, there is an edge in~$\D$ from $u_i$ to $u_{i+1}$ with label $\as(X_i)$. Note that there is always at least one run following~$\as$, which we can obtain by starting at an arbitrary source and following edges until we reach a sink: indeed, condition~(iii) ensures that every node has an outgoing edge with the correct label.
If $\pi$ follows~$\as$, we also say that $\as$ is \emph{consistent} with~$\pi$: note that $\pi$ can be consistent with different assignments if some variables do not occur as the label of any node of~$\pi$. The run $\pi$ is 
\emph{accepting} if the label of the sink at the end of~$\pi$ is $\true$.
An assignment $\as$ is \emph{accepted} by~$\D$, denoted by $\D(\as) = 1$ or $\as \vDash \D$, if there exists an accepting run of~$\D$ following~$\as$; otherwise $\as$ is \emph{rejected} by~$\D$, denoted by $\D(\as) = 0$ or $\as \nvDash \D$.
In this paper, we sometimes refer to a run without explicitly mentioning an assignment that it follows. Finally, observe that $\D$ represents a Boolean function over the set of variables $\X$, and we will often identify $\D$ with the Boolean function that it represents. 

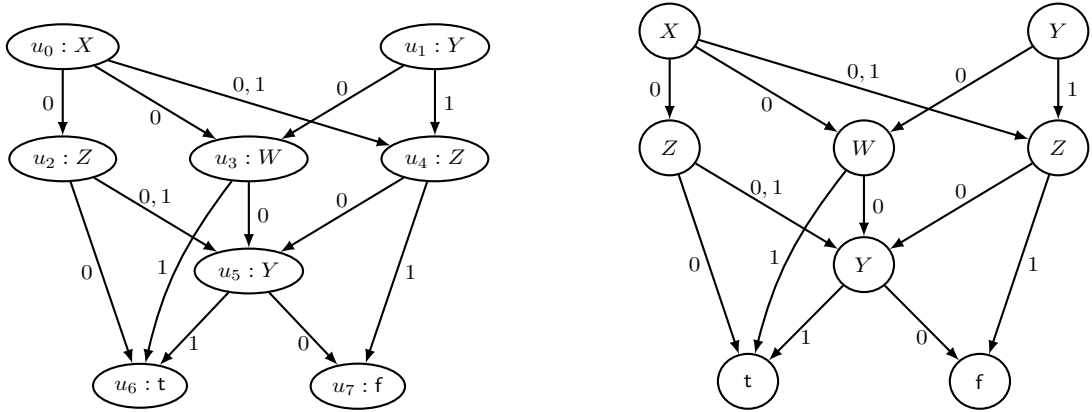
\begin{figure}
\begin{center}
\begin{subfigure}{0.45\textwidth}
\centering
\begin{tikzpicture}[yscale=0.73, xscale=0.80, transform shape]
  \node[thick, draw=black, ellipse, minimum height=8mm, text centered] (n1) {$u_0 : X$};
  \node[thick, draw=black, ellipse, minimum height=8mm, below=12mm of n1, text centered] (n2) {$u_2 : Z$}
    edge[arrin] node[left] {$0$} (n1);
  \node[thick, draw=black, ellipse, minimum height=8mm, right=12mm of n2, text centered] (n3) {$u_3 : W$}
    edge[arrin] node[below] {$0$} (n1);
  \node[thick, draw=black, ellipse, minimum height=8mm, right=12mm of n3, text centered] (n4) {$u_4 : Z$}
    edge[arrin] node[above] {$0,1$} (n1);
  \node[thick, draw=black, ellipse, minimum height=8mm, below=12mm of n3, text centered] (n5) {$u_5 : Y$}
    edge[arrin] node[above] {$0$} (n4)
    edge[arrin] node[above] {$0,1$} (n2)
    edge[arrin] node[right] {$0$} (n3);
  \node[thick, draw=black, ellipse, minimum height=8mm, below left=15mm and 6mm of n5, text centered] (n6) {$u_6 : \true$}
    edge[arrin] node[below] {$1$} (n5)
    edge[arrin, bend left=4mm] node[left] {$1$} (n3)
    edge[arrin] node[left] {$0$} (n2);
  \node[thick, draw=black, ellipse, minimum height=8mm, below right=15mm and 6mm of n5, text centered] (n7) {$u_7 : \false$}
    edge[arrin] node[below] {$0$} (n5)
    edge[arrin] node[right] {$1$} (n4);
  \node[thick, draw=black, ellipse, minimum height=8mm, above=12mm of n4, text centered] (n2) {$u_1 : Y$}
    edge[arrout] node[right] {$1$} (n4)
    edge[arrout] node[above] {$0$} (n3);
\end{tikzpicture}
\caption{A graphical representation of an \nbdd including node identifiers and their labels.}
\label{fig-nbdd-with-id}
\end{subfigure}
\hspace{5mm}
\begin{subfigure}{0.45\textwidth}
\centering
\begin{tikzpicture}[yscale=0.73, xscale=0.80, transform shape]
  \node[circ, minimum size=10mm] (n1) {$X$};
  \node[circ, minimum size=10mm, below=11mm of n1] (n2) {$Z$}
    edge[arrin] node[left] {$0$} (n1);
  \node[circ, minimum size=10mm, right=22mm of n2] (n3) {$W$}
    edge[arrin] node[below] {$0$} (n1);
  \node[circ, minimum size=10mm, right=22mm of n3] (n4) {$Z$}
    edge[arrin] node[above] {$0,1$} (n1);
  \node[circ, minimum size=10mm, below=11mm of n3] (n5) {$Y$}
    edge[arrin] node[above] {$0$} (n4)
    edge[arrin] node[above] {$0,1$} (n2)
    edge[arrin] node[right] {$0$} (n3);
  \node[circ, minimum size=10mm, below left=14mm and 12mm of n5] (n6) {$\true$}
    edge[arrin] node[below] {$1$} (n5)
    edge[arrin, bend left=4mm] node[left] {$1$} (n3)
    edge[arrin] node[left] {$0$} (n2);
  \node[circ, minimum size=10mm, below right=14mm and 12mm of n5] (n7) {$\false$}
    edge[arrin] node[below] {$0$} (n5)
    edge[arrin] node[right] {$1$} (n4);
  \node[circ, minimum size=10mm, above=11mm of n4] (n2) {$Y$}
    edge[arrout] node[above] {$0$} (n3)
    edge[arrout] node[right] {$1$} (n4);
\end{tikzpicture}
\caption{The usual graphical representation of an \nbdd, where node identifiers are not included.}
\label{fig-nbdd-without-id}
\end{subfigure}
\end{center}
\caption{An \nbdd over the set of variables $\X = \{X,Y,Z,W\}$. \label{fig-nbdd}}
\end{figure}

\begin{example}
  An \nbdd $\D$ over the set of variables $\X = \{X,Y,Z,W\}$ is shown in Figure~\ref{fig-nbdd-with-id}. For each node we include its identifier and label; for instance, $u_0 : X$ indicates that node $u_0$ has label~$X$. 
  Moreover, we depict edges with their labels; for instance, $(u_0, 0, u_2)$ is the only edge from $u_0$ to $u_2$, while $(u_0, 0, u_4)$, $(u_0, 1, u_4)$ are the two edges from $u_0$ to $u_4$ (the symbol $0,1$ next to the edge from $u_0$ to $u_4$ is used to denote two edges, one with label $0$ and the other one with label $1$).
  The two sources
  of $\D$ are $u_0$ and $u_1$, while the two sinks of $\D$ are $u_6$ and $u_7$.

  Consider the assignment $\as_1 \in \assign(\X)$ such that $\as_1(X) = \as_1(Y) = \as_1(Z) = \as_1(W) = 0$. Then we have that $\as_1$ is consistent with the run $\pi_1 = u_0, u_2, u_6$. Notice that an assignment can be consistent with many different runs of an \nbdd; for instance, $\as_1$ is consistent with $\pi_1$ as well as with the run $\pi_2 = u_1, u_3, u_5, u_7$. An assignment is accepted by an \nbdd if there exists at least one accepting run of the \nbdd that is consistent with it; for instance $\D(\as_1) = 1$ since $\as_1$ is consistent with the accepting run $\pi_1$ $(\D(\as_1) = 1$ despite the fact that $\as_1$ is consistent with the run $\pi_2 = u_1, u_3, u_5, u_7$ and the label of sink $u_7$ is $\false)$. 

We depict in Figure \ref{fig-nbdd-without-id} the same \nbdd as in Figure \ref{fig-nbdd-with-id}, but without including node identifiers. The usual graphical representation of an \nbdd is the one given in  Figure \ref{fig-nbdd-without-id}.
\end{example}

\noindent
Nondeterministic binary decision diagrams are classified according to two dimensions.
\begin{itemize}
  \item {\bf Variable structuredness (free, ordered).} Let $\D$ be an \nbdd over a set of variables $\X$. Then $\D$ is \emph{free} (\nfbdd) if, for every run $\pi$ of $\D$, no two distinct nodes in $\pi$ have the same label. In addition, $\D$ is {\em ordered} (\nobdd) if there exists a linear order $<$ on the set $\X$ such that, if a node $u_1$ appears before a node $u_2$ in some run of $\D$, then, letting $X_1 \in \X$ be the label of $u_1$ and $X_2 \in \X$ the label of $u_2$,
    we have 
    $X_1 < X_2$. Notice that an $\nobdd$ is in particular an $\nfbdd$.

    The notions of $\nfbdd$ and $\nobdd$ are defined in terms of the runs of an $\nbdd$. Thanks to condition~(iii) of the definition of an $\nbdd$, such notions can be equivalently defined in terms of the paths of an $\nbdd$. In particular, it would be equivalent to say that an $\nbdd$ $\D$ is free if for every (directed) path $\pi$ in $\D$, no two distinct nodes in $\pi$ have the same label, and likewise for the condition that the \nbdd is ordered.%

\item {\bf Ambiguity level.} Let $\D$ be an \nbdd over a set of variables $\X$. Then $\D$ is {\em unambiguous} (\ubdd) if, for every assignment $\as \in \assign(\X)$, there exists at most one accepting run of $\D$ that is consistent with $\as$. Moreover, $\D$ is {\em deterministic}, which is referred to as \bdd \citep{L59,W04}), if, for every assignment $\as \in \assign(\X)$, there exists exactly one run of~$\D$ that is consistent with $\as$.

As mentioned before, thanks to condition~(iii) of the definition of an \nbdd,
for every \nbdd~$\D$ over a set of variables $\X$ and every $\as \in \assign(\X)$, there exists at least one run of $\D$ that is consistent with $\as$.
 In the case where $\D$ is deterministic, such a run must be unique. This leads to the following equivalent definition of a \bdd, which is the most commonly used in the literature: $\D$ is a \bdd if and only if 
$\D$ has a single source,
which is called the \emph{root} of $\D$, and every internal node of $\D$ has exactly one outgoing $0$-edge and exactly one outgoing $1$-edge.
\end{itemize}

\begin{table}
\caption{Classification of nondeterministic binary decision diagrams based on ambiguity level (nondeterministic, unambiguous, or deterministic) and variable structuredness (unrestricted, free, or ordered).\label{tab-bdd}}
\begin{center}
\begin{tabular}{lccc}\toprule
  & \bfseries Unrestricted & \bfseries Free & \bfseries Ordered\\\midrule
 \bfseries Nondeterministic & \nbdd & \nfbdd & \nobdd\\
 \bfseries Unambiguous & \ubdd & \ufbdd & \uobdd\\
 \bfseries Deterministic & \bdd & \fbdd & \obdd\\\bottomrule
\end{tabular}
\end{center}
\end{table}

The combination of the previous two dimensions gives rises to 9 different classes of nondeterministic binary decision diagrams, which are shown in Table~\ref{tab-bdd}.
The most widely used models among these classes are the deterministic variants, i.e., binary decision diagrams (\bdd) \citep{L59}, free binary decision diagrams (\fbdd) \citep{FHS78,BCW80}, and ordered binary decision diagrams (\obdd) \citep{B86}. An \fbdd is also referred to as a \emph{read-once branching program} in the literature, where the term \emph{nondeterministic read-once branching program}~\citep{razgon2014no} is used to denote\footnote{Note that \citep{razgon2014no} also defines a notion of \nfbdd, or \emph{normalized FBDD}, which is different from the \nfbdds that we consider.} \nfbdds.

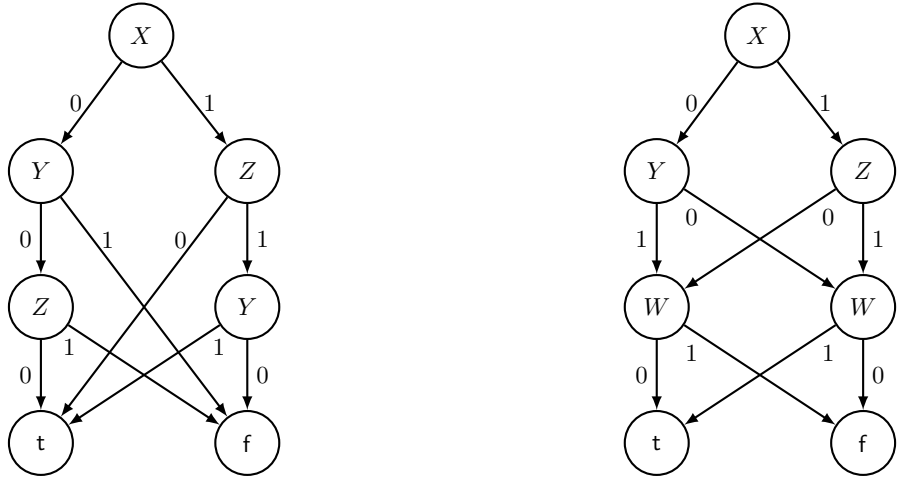
\begin{figure}
\begin{center}
\begin{subfigure}{0.45\textwidth}
\centering
\begin{tikzpicture}[scale=0.85, transform shape]
  \node[minimum size=10mm] (n1) {};
  \node[circ, right=5.5mm of n1, minimum size=10mm] (n0) {$X$};
  \node[circ, minimum size=10mm, below=11mm of n1] (n3) {$Y$}
      edge[arrin] node[left] {$0$} (n0);
  \node[circ, minimum size=10mm, right=22mm of n3] (n4) {$Z$}
    edge[arrin] node[right] {$1$} (n0);
  \node[circ, minimum size=10mm, below=11mm of n3] (n5) {$Z$}
    edge[arrin] node[left] {$0$} (n3);
  \node[circ, minimum size=10mm, below=11mm of n4] (n9) {$Y$}
    edge[arrin] node[right] {$1$} (n4);

\node[circ, minimum size=10mm, below=11mm of n5] (n6) {$\true$}
    edge[arrin] node[left] {} (n4)
    edge[arrin] node[below] {} (n9)
    edge[arrin] node[left] {$0$} (n5);
  \node[circ, minimum size=10mm, below=11mm of n9] (n7) {$\false$}
    edge[arrin] node[right] {} (n3)
    edge[arrin] node[right] {$0$} (n9)
    edge[arrin] node[below] {} (n5);

\node[below right= 4.5mm and 4.5mm of n3] (e39) {$1$};

\node[below left= 4.5mm and 4.5mm of n4] (e48) {$0$};

\node[below left= 0mm and -1.5mm of n9] (e9f) {$1$};

\node[below right= 0mm and -1.5mm of n5] (e5f) {$1$};
\end{tikzpicture}
\caption{A free binary
decision diagram.}
\label{fig-fbdd}
\end{subfigure}
\hspace{5mm}
\begin{subfigure}{0.45\textwidth}
\centering
\begin{tikzpicture}[scale=0.85, transform shape]
  \node[minimum size=10mm] (n1) {};
  \node[circ, right=5.5mm of n1, minimum size=10mm] (n0) {$X$};
  \node[circ, minimum size=10mm, below=11mm of n1] (n3) {$Y$}
      edge[arrin] node[left] {$0$} (n0);
  \node[circ, minimum size=10mm, right=22mm of n3] (n4) {$Z$}
    edge[arrin] node[right] {$1$} (n0);
  \node[circ, minimum size=10mm, below=11mm of n3] (n5) {$W$}
    edge[arrin] node[left] {} (n4)
    edge[arrin] node[left] {$1$} (n3);
  \node[circ, minimum size=10mm, below=11mm of n4] (n9) {$W$}
  edge[arrin] node[right] {} (n3)
    edge[arrin] node[right] {$1$} (n4);

\node[circ, minimum size=10mm, below=11mm of n5] (n6) {$\true$}
    edge[arrin] node[below] {} (n9)
    edge[arrin] node[left] {$0$} (n5);
  \node[circ, minimum size=10mm, below=11mm of n9] (n7) {$\false$}
    edge[arrin] node[right] {$0$} (n9)
    edge[arrin] node[below] {} (n5);

\node[below right= 1mm and -0.5mm of n3] (e39) {$0$};

\node[below left= 1mm and -0.5mm of n4] (e48) {$0$};

\node[below left= 1mm and -0.5mm of n9] (e9f) {$1$};

\node[below right= 1mm and -0.5mm of n5] (e5f) {$1$};
\end{tikzpicture}
\caption{An ordered binary
decision diagram.}
\label{fig-obdd}
\end{subfigure}
\end{center}
\caption{Two deterministic binary decision diagrams over the set of variables $\X = \{X,Y,Z,W\}$. \label{fig-fbdd-obdd}}
\end{figure}

\begin{example}
The \nbdd in Figure \ref{fig-nbdd} is not free: indeed, for the run $u_1, u_3, u_5, u_7$, the nodes $u_1$ and $u_5$ have the same label $Y$. On the other hand, the \nbdd shown in Figure \ref{fig-fbdd} is free and deterministic, so it is an \fbdd. 
Moreover, the \nbdd shown in Figure \ref{fig-obdd} is ordered and deterministic, because every run in it follows the linear order $X < Y < Z < W$, so it is an \obdd. The \nbdd of Figure~\ref{fig-fbdd}, on the other hand, is not ordered.
\end{example}

\subsection{Binary decision diagrams with or-nodes}
Nondeterministic binary decision
diagrams have also been defined by extending binary decision diagrams with {\em
or-nodes} \citep{ACMS20,DBSSWCH10} (they are also called {\em guessing
nodes}~\citep{razborov1991lower}).
An or-node $u$ in a \bdd is an internal node such that $u$ is labeled with the disjunction symbol $\vee$ instead of a variable, and the outgoing edges of $u$ are not labeled. Acceptance of a \bdd with or-nodes is defined as for the case of \bdds; in particular, if a run goes through an or-node, then it chooses an outgoing edge of this or-node, which does not impose any restriction on the values assigned to variables. In this sense, or-nodes can be used to encode nondeterministic choices as shown in the following example:
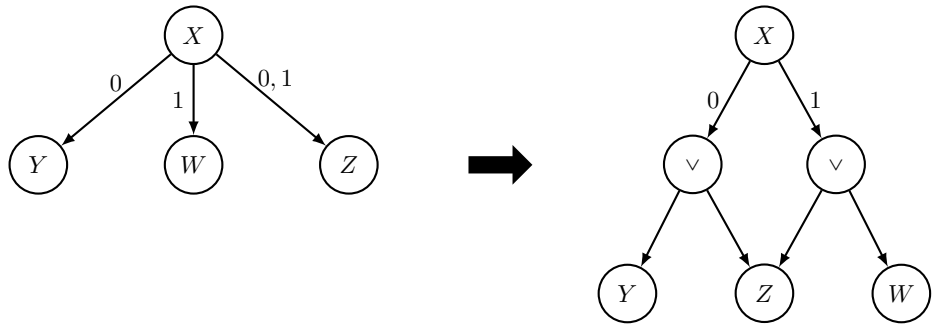
\begin{center}
\begin{tikzpicture}[scale=0.85, transform shape]
  \node[circ, minimum size=9mm] (n1) {$X$};
  \node[circ, minimum size=9mm, below=11mm of n1] (n3) {$W$}
    edge[arrin] node[left] {$1$} (n1);
  \node[circ, minimum size=9mm, left=15mm of n3] (n2) {$Y$}
    edge[arrin] node[above] {$0$} (n1);
  \node[circ, minimum size=9mm, right=15mm of n3] (n4) {$Z$}
    edge[arrin] node[right, above] {\,\,$0,1$} (n1);
  
\node[single arrow, fill=black, right=14mm of n4, minimum width = 10pt, single arrow head extend=5pt, minimum height=10mm] {}; 
      
  \node[circ, minimum size=9mm, right = 80mm of n1] (tn1) {$X$};
  \node[minimum size=9mm, below=11mm of tn1] (tn2) {};
  \node[circ, minimum size=9mm, left=2mm of tn2] (tn3) {$\vee$}
    edge[arrin] node[left] {$0$} (tn1);
  \node[circ, minimum size=9mm, right=2mm of tn2] (tn4) {$\vee$}
    edge[arrin] node[right] {$1$} (tn1);
  \node[minimum size=9mm, below=11mm of tn3] (tn6) {};
  \node[circ, minimum size=9mm, left=1mm of tn6] (tn8) {$Y$}
    edge[arrin] node[above] {} (tn3);   
  \node[minimum size=9mm, below=11mm of tn4] (tn7) {};
  \node[circ, minimum size=9mm, right=1mm of tn7] (tn10) {$W$}
    edge[arrin] node[above] {} (tn4);
  \node[circ, minimum size=9mm, below=31.1mm of tn1] (tn9) {$Z$}
    edge[arrin] node[above] {} (tn3)
    edge[arrin] node[above] {} (tn4); 

\end{tikzpicture}
\end{center}
Part of an \nbdd is shown in the left-hand side of this figure, while its representation as a \bdd with or-nodes is shown in the right-hand side. In particular, if the variable $X$ is assigned value $0$ in the \nbdd, then a run can move either to the node with label $Y$ or to the node with label $Z$. Such a choice is represented in the \bdd by connecting the outgoing $0$-edge of variable $X$ to an or-node, which in turn is connected (by means of edges without labels) to the nodes with labels $Y$ and $Z$; in this way, we indicate that if $X$ is assigned value $0$ in a run of the \bdd, then this run must move to an or-node, from which it must choose whether to move either to variable $Y$ or to variable $Z$. It is straightforward to see that this idea can be used to translate in polynomial time an \nbdd into an equivalent \bdd with or-nodes.

In the other direction, a \bdd with or-nodes can be translated in polynomial time 
into an equivalent \nbdd by applying the transformation shown in Figure~\ref{fig-or-to-nbdd}, which we explain next. Part of a \bdd with or-nodes is shown in the left-hand side of this figure, while its representation as an \nbdd is shown in the right-hand side.
More precisely, for a pair of nodes $u$, $v$ that are labeled by variables in a \bdd with or-nodes $\D$, an 
\emph{or-path} from $u$ to $v$ is a path~$\pi$ from $u$ to $v$ in~$\D$ such that every node in~$\pi$ except for $u$ and $v$ is an or-node. For example, the following are or-paths from the node $X$ to the node $Z$ in Figure~\ref{fig-or-to-nbdd}:
\begin{center}
\begin{tikzpicture}[scale=0.85, transform shape]
  \node[circ, minimum size=9mm] (n1) {$X$};
  \node[minimum size=9mm, below=9mm of n1] (aux) {};
  \node[circ, minimum size=9mm, left=3mm of aux] (n2) {$\vee$}
    edge[arrin] node[left] (e) {$0$} (n1);
  \node[circ, minimum size=9mm, below=9mm of aux] (n4) {$Z$}
    edge[arrin] node[right] {} (n2);

  \node[circ, right = 20mm of n1, minimum size=9mm] (sn1) {$X$};
  \node[minimum size=9mm, below=9mm of sn1] (aux) {};
  \node[circ, minimum size=9mm, right=3mm of aux] (sn3) {$\vee$}
    edge[arrin] node[right] {$1$} (sn1);
  \node[circ, minimum size=9mm, below=9mm of aux] (sn4) {$Z$}
    edge[arrin] node[right] {} (sn3);

\end{tikzpicture}
\end{center}
Notice that the first edge in an or-path must be labeled $0$ or $1$; in the
first case, the or-path is said to be a \emph{$0$-or-path}, while in the second is said
to be a \emph{$1$-or-path}. Moreover, given  a node $u$ that is labeled by a variable,
the \emph{or-closure} of $u$ is defined as the set of nodes $v$ such that $v$ is
labeled by a variable and there exists an or-path from $u$ to $v$. In the \bdd
with or-nodes in Figure~\ref{fig-or-to-nbdd}, the or-closure of $X$ consists of
the nodes with labels $Y$, $Z$, $W$ and $V$. Then in the translation of a \bdd
with or-nodes into an \nbdd, for every node $v$ in the or-closure of a node $u$,
a $0$-edge from $u$ to $v$ is included in the \nbdd if there exists a
$0$-or-path from $u$ to $v$, and a $1$-edge from $u$ to $v$ is included in the
\nbdd if there exists a $1$-or-path from $u$ to $v$. An example of such a
transformation is shown in Figure~\ref{fig-or-to-nbdd}. It is straightforward to
see
that this idea can be used to translate in polynomial time any  \bdd with
or-nodes~$\D$ into an equivalent \nbdd $\D'$.\footnote{Note that, if the root of~$\D$ is an or-node, then $\D'$ will have multiple
sources.}
Observe that the complexity of that operation is essentially that of computing the transitive closure 
of the binary decision diagram, which can be done in polynomial time.
Furthermore, there is no
need to compute this transitive closure
if no or-node is connected to another or-node: if we impose this condition on the input, then the translation becomes linear-time.

\begin{figure}
\begin{center}
\begin{tikzpicture}[scale=0.85, transform shape]
  \node[circ, minimum size=9mm] (n1) {$X$};
  \node[minimum size=9mm, below=9mm of n1] (aux) {};
  \node[circ, minimum size=9mm, left=3mm of aux] (n2) {$\vee$}
    edge[arrin] node[left] (e) {$0$} (n1);
  \node[circ, minimum size=9mm, right=3mm of aux] (n3) {$\vee$}
    edge[arrin] node[right] {$1$} (n1);
  \node[circ, minimum size=9mm, below=9mm of aux] (n4) {$Z$}
    edge[arrin] node[right] {} (n2)
    edge[arrin] node[right] {} (n3);
  \node[circ, minimum size=9mm, left=16mm of n4] (n5) {$Y$}
    edge[arrin] node[right] {} (n2);
  \node[circ, minimum size=9mm, right=16mm of n4] (n6) {$\vee$}
    edge[arrin] node[right] {} (n3);
  \node[minimum size=9mm, below=9mm of n6] (aux6) {};
  \node[circ, minimum size=9mm, right=4mm of aux6] (n7) {$V$}
    edge[arrin] node[right] {} (n6);
  \node[circ, minimum size=9mm, left=4mm of aux6] (n8) {$W$}
    edge[bend left = 12mm, arrin] node[right] {} (n2)
    edge[arrin] node[right] {} (n6);

\node[single arrow, fill=black, right=57mm of e, minimum width = 10pt, single arrow head extend=5pt, minimum height=10mm] {}; 

\node[circ, minimum size=9mm, right = 100mm of n1] (tn1) {$X$};
\node[minimum size=9mm, below=9mm of tn1] (aux) {};
\node[circ, minimum size=9mm, left=-2mm of aux] (tn2) {$Z$}
    edge[arrin] node[left] {$0,1$} (tn1);
\node[circ, minimum size=9mm, left=20mm of aux] (tn3) {$Y$}
    edge[arrin] node[above left] {$0$} (tn1);
\node[circ, minimum size=9mm, right=-2mm of aux] (tn4) {$W$}
    edge[arrin] node[right] {$0,1$} (tn1);
\node[circ, minimum size=9mm, right=20mm of aux] (tn5) {$V$}
    edge[arrin] node[above right] {$1$} (tn1);

\end{tikzpicture}
\end{center}
\caption{Transformation of a \bdd with or-nodes into an \nbdd.\label{fig-or-to-nbdd}}
\end{figure}
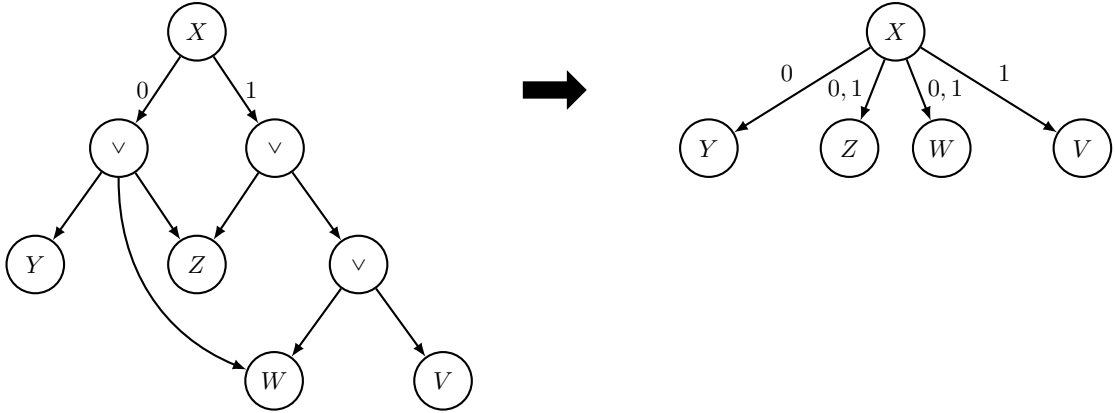

The different conditions on variable structuredness and ambiguity level for \nbdds can be directly extended to \bdds with or-nodes. Thus, for example, we can talk about \fbdds with or-nodes and \obdds with or nodes. The procedures to transform \nbdds into \bdds with or-nodes and 
vice-versa preserve all such conditions. In this sense, 
up to the translation that we presented, 
the models of \nbdds and \bdds with or-nodes are completely interchangeable:
we use \nbdds in the sequel.
We last note that some restricted cases of \bdds with or-nodes have been considered in the literature, e.g., \pbdds, which are (nondeterministic) disjunctions of \obdds with different orders~\citep{BW97}. Again, these models can be represented within the framework based on \nbdds that is used in this paper.

\subsection{Completion and zero-suppressed semantics}

We present in this section the \emph{completion} transformation on \bdds, and the alternative semantics of \bdds called the \emph{zero-suppressed semantics}.

\paragraph*{Completion.}
The notion of \emph{complete} \nbdds has been studied, e.g., in
\citep{bollig2016minimization}.
An \nbdd is \emph{complete} if all variables are tested on all
runs. More precisely, for each run $\pi = u_1, u_2, \ldots u_k$, the set of
variables that occur as labels of $\{u_1, u_2, \ldots, u_k\}$ is equal to the
set $\X$ of all variables. In particular, if the binary decision diagram is
free, then every such path must consist of exactly $|\X|$ internal nodes
followed by a sink. Further, if the binary decision diagram is ordered, the
sequence of variables tested by every run is exactly the linear order on the
variables that the binary decision diagram follows.

Completeness can be useful for certain tasks. Consider for instance the counting
problem for \fbdds: the input is a \fbdd $\mathcal{D}$ over a set of
variables $\X$, and the output is $\ssat(\mathcal{D})$, i.e., the number of assignments of $\assign(\X)$
that satisfy $\mathcal{D}$. If $\mathcal{D}$ is complete, then the following linear-time bottom-up algorithm is
correct. Annotate all $\true$-sinks with $1$ and all $0$-sinks with $0$. Then for an internal node $n$,
annotate it by the sum of the annotations of the nodes that can be reached from~$n$ by following only one edge,
then output the sum of the annotations of all source nodes. If $\mathcal{D}$ was not complete, this would not
return the correct result, as this might have missed some satisfying assignments.

We can always complete an \nbdd in polynomial time by adding extra internal nodes before sinks. More precisely, assuming $\X = \{X_1, \ldots, X_k\}$, we replace each sink $u$ by a sequence of nodes $u_1, \ldots, u_k, u_{k+1}$ where the label of $u_i$ is $X_i$ for every $i \in [1,k]$, the label of $u_{k+1}$ is the same as the label of $u$, and there is a $0$-edge and a $1$-edge from $u_i$ to $u_{i+1}$ for every $i \in [1,k]$. Notice that this translation does not affect unambiguity or determinism; however, the resulting \nbdd is in general not free or ordered. The complexity of the procedure is $O(|\D| \cdot |\X|)$, where $|\D|$ is the size of the input \nbdd $\D$ and $\X$ is the set of variables. 

More interestingly, we can always complete an \nfbdd in polynomial time while ensuring that the result is still an \nfbdd, as shown in~\citep{wegener2000branching} (see also~\citep{ACMS20}). 
To do this, we rewrite the \nfbdd from the sources
to the sinks in polynomial time while ensuring that, for every node $u$, all paths from a source 
to~$u$ test the same set $S_u$ of variables, and every variable in $S_u$ is tested exactly once in each such path. The base case is that of a source 
$u$ labeled with variable $X$ which tests the set $S_u = \{X\}$. Inductively, considering an internal node $u$ labeled with~$X$ with incoming edges from nodes $u_1, \ldots, u_\ell$ testing sets $S_{u_1}, \ldots, S_{u_\ell}$, and letting $S = \bigcup_i S_{u_i}$, we replace each edge from $u_i$ to $u$ by a sequence of nodes like in the previous paragraph which test the variables of the (possibly empty) set $S \smallsetminus S_i$. By inductive assumption, all paths reaching the end of such a sequence are testing precisely $S$. As the binary decision diagram is an \nfbdd, we know that the label $X$ of node $u$ does not belong to $S_i$ for every~$i \in [1, \ell]$. Hence, $x \notin S$ and all paths reaching node $u$ after the rewriting are indeed testing all variables of $S_u = S \cup \{X\}$ precisely once. Last, for sinks, we proceed as in the previous paragraph to make sure that any remaining variables are tested. Note that this translation again does not affect unambiguity or determinism, and the result is complete and is still free, and again the complexity is $O(|\D| \cdot |\X|)$ for $\D$ the input \nfbdd and $\X$ its set of variables.

We last observe that this process on \nfbdds, when applied to \nobdds, can be performed in a way that gives an \nobdd as a result, if the sequences of variables that we insert always follow the order~$<$ of the \nobdd. For this we can observe that the sets $S_n$ of the previous proofs are always a prefix of the order~$<$. Again the transformation does not affect unambiguity or determinism, and the complexity is the same. In this specific case, a simpler description of this algorithm is then the following.
Assume that the \nobdd conforms to the linear order $X_1 < X_2 < \ldots < X_k$. In order to complete it, we can replace each edge $X_i \xrightarrow{a} X_j$ by the sequence of edges
$X_i \xrightarrow{a} X_{i+1} \xrightarrow{0,1} X_{i+2} \xrightarrow{0,1} \cdots \xrightarrow{0,1} X_j$, where the intermediate steps represent fresh nodes that test the indicated variables.

\paragraph*{Zero-suppressed semantics.}
An alternative way to define the semantics of an \nbdd is via the so-called \emph{zero-suppressed semantics}, see \citep[Section 8.1]{wegener2000branching} or~\citep{minato1993zero}. Intuitively, in a zero-suppressed \nbdd, it is assumed that every variable not mentioned along a run %
takes the value $0$, unlike the standard semantics presented so far where the value of such variables is unconstrained.
More formally, given a set of variables $\X$, a zero-suppressed \nbdd $\D$ over $\X$ is defined like an \nbdd, but with the following modification in the definition of the acceptance condition for~$\D$. Let $\as \in \assign(\X)$ and $\pi = u_1, u_2, \ldots u_k$ be a 
path from a source to a sink of
$\D$, and let $X_i$ be the label of $u_i$ for each $i \in [1, k-1]$. Then $\pi$ is said to be \emph{a run consistent with $\as$} if there exists an edge from $u_i$ to $u_{i+1}$ with label $\as(X_i)$ for each $i \in[1, k-1]$, and $\as(X) = 0$ for each variable $X \in \X \smallsetminus \{X_1, \ldots, X_{k-1}\}$.

The purpose of zero-suppressed semantics is that, in practice, it can yield smaller binary decision diagrams. For example, a Boolean function over a set of variables $\X$ with the only satisfying assignment being the one that assigns $0$ to all variables can be represented in the zero-suppressed semantics simply as a single source node that is also a $\true$-sink. In contrast, using the usual semantics requires $|X|$ nodes.
However, as we explain next, there are simple polynomial-time conversions between
the two semantics.

For an \nbdd that is complete, 
given that all runs must test every variable, there is no difference between the standard semantics and the zero-suppressed semantics. For this reason, 
every \nbdd can be converted in polynomial time to an equivalent zero-suppressed \nbdd, which preserves the variable structuredness and ambiguity level of the \nbdd being translated. 
Conversely, given a zero-suppressed \nbdd $\D$, it can be transformed into an equivalent \nbdd by considering a variant of the completion procedure presented before: in that procedure, we only add nodes that check that the missing variables are mapped to $0$.
This procedure makes $\D$ complete without changing the function that it represents in the zero-suppressed semantics. 
In addition,
it preserves the variable structuredness and ambiguity level of the zero-suppressed \nbdd being translated.

\subsection{Decision trees}
If we disallow sharing (i.e., we disallow internal nodes having multiple incoming edges), we get decision trees. More precisely,
an \nbdd $\D$ is a \emph{nondeterministic decision forest} (\ndf) 
if, ignoring the sinks, the underlying graph of $\D$ is a tree.
If, in addition, it has only one source node, then it is a \emph{nondeterministic decision tree} (\ndt). 
This notion of \emph{nondeterministic decision trees} has already been studied, e.g., in~\citep[Theorem~10.1.4]{wegener2000branching}.
For example, the \nbdd in Figure \ref{fig-fbdd} is an \ndt, which is usually depicted in the following way, implicitly indicating that leaves are not considered in the underlying structure that defines an \ndt:

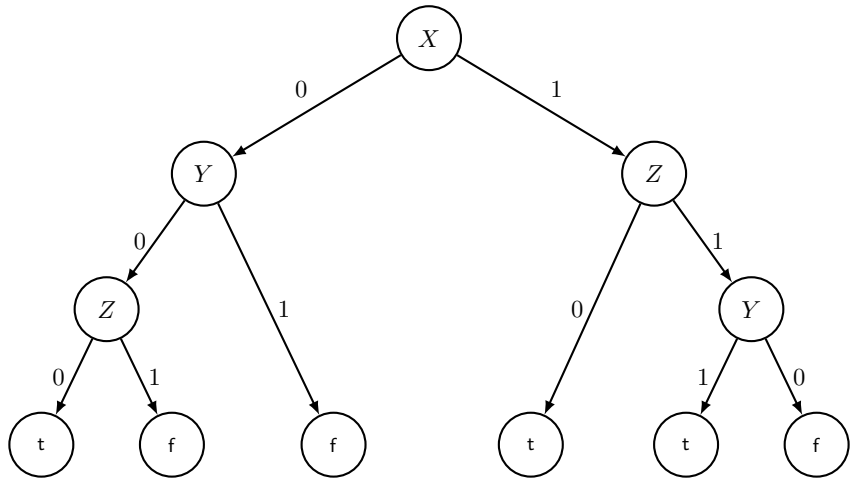
\begin{center}
\begin{tikzpicture}[scale=0.85, transform shape]
  \node[circ, minimum size=10mm] (n0) {$X$};
\node[minimum size=10mm, below=11mm of n0] (a0) {};
  \node[circ, minimum size=10mm, left=25mm of a0] (n3) {$Y$}
      edge[arrin] node[above left] {$0$} (n0);
  \node[circ, minimum size=10mm, right=25mm of a0] (n4) {$Z$}
    edge[arrin] node[above right] {$1$} (n0);
  \node[minimum size=10mm, below=11mm of n3] (a1) {};
  \node[circ, minimum size=10mm, left=5mm of a1] (n5) {$Z$}
    edge[arrin] node[left] {$0$} (n3);
\node[minimum size=10mm, below=11mm of n4] (a4) {};
  \node[circ, minimum size=10mm, right=5mm of a4] (n9) {$Y$}
    edge[arrin] node[right] {$1$} (n4);

\node[minimum size=10mm, below=11mm of n5] (a2) {};
\node[minimum size=10mm, below=11mm of n9] (a3) {};

\node[circ, minimum size=10mm, left=0mm of a2] (n6) {\strue}
    edge[arrin] node[left] {$0$} (n5);
  \node[circ, minimum size=10mm, right=0mm of a2] (n7) {\sfalse}
    edge[arrin] node[right] {$1$} (n5);
  \node[circ, minimum size=10mm, right=15mm of n7] (n7) {\sfalse}
    edge[arrin] node[right] {$1$} (n3);

\node[circ, minimum size=10mm, left=0mm of a3] (nn1) {\strue}
    edge[arrin] node[left] {$1$} (n9);
\node[circ, minimum size=10mm, left=14mm of nn1] (nn2) {\strue}
    edge[arrin] node[left] {$0$} (n4);
\node[circ, minimum size=10mm, right=0mm of a3] (nn3) {\sfalse}
        edge[arrin] node[right] {$0$} (n9);

\end{tikzpicture}
\end{center}
On the other hand, the \nbdd in Figure \ref{fig-obdd} is not an \ndt. 

As a special case of \nbdds, \ndfs and \ndts inherit property definitions such as variable structuredness and ambiguity levels. As with binary decision diagrams, the most widely used subclass are the \ndts that are deterministic, which are usually referred to simply as \emph{decision trees} (\ddt) \citep{SL91,WKQGYMMNLYZSHS08}. 
For example, the \nbdd in Figure \ref{fig-fbdd} is a \ddt, which can be easily verified considering the previous figure.

It is important to mention that \ndfs and \ddts are usually assumed to be free because, unlike general binary decision diagrams, an \ndf can always be transformed 
in polynomial time
into an equivalent free \ndf (and likewise for \ddts). More precisely, for every node $u$ of an \ndf $\D$, there exists a unique path from the corresponding source of $\D$ to $u$. The labels of the vertices and edges traversed on this path specify a partial assignment for a set of variables.
If the node $u$ tests a variable $X$ that was already tested on this path, then the value $a$ of~$X$ is already specified in the partial assignment. Thus, we can remove~$u$, re-connect the outgoing edges from $u$ labeled by~$a$ to the predecessor of~$u$ in the path, and remove the other outgoing edges from $u$ (which are labeled $1 - a$). Repeating this process in a traversal of $\D$ makes this \ndf free. Observe that this process does not work with general binary decision diagrams as there may be different paths reaching a node $u$ with conflicting values for the variable tested by~$u$.

\section{Boolean Circuits} \label{sec:circuits}
We now move from binary decision diagrams to the more general formalism of
Boolean circuits. The section is structured similarly to the previous section:
we give the basic definitions of a Boolean circuit and its semantics, and define notions
of structuredness and ambiguity levels. We then mention the definition of
Sentential Decision Diagrams (SDDs) as a special case of Boolean circuits. We then
mention the case of Boolean circuits without sharing (aka formulas), and the notion of
smoothness for Boolean circuits (corresponding to completeness for binary
decision diagrams).

\subsection{Definitions}

A \emph{circuit $C$ over variables $\X$} is a rooted $\{\land,\lor,\lnot,\true,\false\}\cup
\X$-node labeled DAG $C = (N, W, \lambda)$ together with a designated sink
$g_0\in N$ called the \emph{output gate of $C$}. The vertices $N$ are called
\emph{gates}, the edges $W$ are called \emph{wires}. We often abuse notation
and write $g\in C$ to mean $g\in N$. An \emph{input} of a gate $g \in C$ is a
gate $g' \in C$ that has a wire to~$g$, i.e., we have $(g', g) \in W$. The
gates of~$C$ can be of several kinds:
\begin{itemize}
  \item If $\lambda(g)= \true$, then $g$ is a \emph{constant true-gate} (also sometimes called \emph{constant $1$-gate}), and it must then have no input;
\item If $\lambda(g)= \false$, then $g$ is a \emph{constant false-gate} (also sometimes called \emph{constant $0$-gate}), and it must then have no input;
\item If $\lambda(g)\in \X$, then $g$ is a \emph{variable gate}, and it must then have no input;
\item If $\lambda(g) = \lnot$, then $g$ is a \emph{negation gate}, and it must then have exactly one input;
\item If $\lambda(g) = \lor$ (resp, $\lambda(g) = \land$), then $g$ is a \emph{$\land$-gate} (resp., \emph{$\lor$-gate}), and it must then have at least one input.
\end{itemize}

The Boolean circuit $C$ represents a Boolean function over $\X$ in
the following way. Given an assignment $\ba\colon \X \to \{0, 1\}$, we extend it
to give a Boolean value to all gates of the Boolean circuit by bottom-up induction:
\begin{itemize}
  \item The value of a constant true-gate $g$ is $\ba(g) = 1$;
  \item The value of a constant false-gate $g$ is $\ba(g) = 0$;
    \item The value of a variable gate $g$ annotated with variable $X$ is $\ba(g) \colonequals \ba(X)$;
    \item The value of a negation gate $g$ is $\ba(g) \colonequals 1 - \ba(g')$, where $g'$ is the input of~$g$;
    \item The value of an $\land$-gate $g$ (resp., $\lor$-gate $g$) with input gates $g_1, \ldots, g_n$ is $\ba(g) \colonequals \bigwedge_i \ba(g_i)$ (resp., $\ba(g) \colonequals \bigvee_i \ba(g_i)$).
\end{itemize}
The Boolean function defined by $C$ is then the one that maps assignments $\ba$
to the value $\ba(g_0)$ of the output gate of the Boolean circuit. We sometimes abuse
terminology and call the variable gates the \emph{inputs} of the Boolean circuit.

The size~$|C|$ of a Boolean circuit is its number of wires.  For a gate~$g$ of~$C$, we
denote by~$\vars(g)$ the set of variables that have a directed path to~$g$
in~$C$, and we denote by~$C_g$ the Boolean circuit over~$\X$ whose output gate
is~$g$. 

In the rest of this document, we will always consider Boolean circuits in
\emph{negation normal form} (NNF), where negation gates are always applied to
variables; formally, for any negation gate $g$, then its one input must be a
variable gate. We can equivalently see NNF Boolean circuits as positive Boolean circuits
(circuits without negation) defined directly on the literals (i.e., the
variables and their negations).\\

Just like binary decision diagrams, we classify Boolean circuits according to two
dimensions.

\begin{itemize}
    \item \textbf{Variable structuredness.} A Boolean circuit is called
      \emph{decomposable} if, intuitively, $\land$-gates partition the
      variables into disjoint sets. Formally, an $\land$-gate~$g$ of~$C$ is
      \emph{decomposable} if it has exactly two input gates $g_1\neq g_2$ such
      that we have~$\vars(g_1) \cap \vars(g_2) = \emptyset$. A Boolean circuit $C$ is
      \emph{decomposable} if every $\land$-gate of~$C$ is. Note that this is a
      syntactic condition that can easily be checked in time $O(|C|\cdot |\X|)$. A
      decomposable NNF Boolean circuit is called a \emph{DNNF}. We point out that DNNFs
      are sometimes defined without the restriction that $\land$-nodes always
      have two inputs; we impose this for convenience, and this is without much
      loss of generality as this can be enforced in linear time (using constant gates).
    
    Further, a DNNF is \emph{structured} if the partitions that are defined by the
    $\land$-gates are compatible, in the following sense. A \emph{v-tree} over
    the set of variables $\X$ is a rooted full binary tree $T$ whose leaves are
    in bijection with $\X$.  We always identify each leaf with the associated
    element of~$\X$. For a node~$n\in T$, we abuse notation and denote
    by~$\vars(n)$ the set of variables in the subtree rooted at~$n$. A DNNF $D$
    is \emph{structured by the v-tree $T$} if there exists a mapping $\rho$
    labeling each $\land$-gate of $g$ with a node of $T$ that satisfies the
    following: for every $\land$-gate $g$ of $D$ with two inputs $g_1, g_2$,
    the node $\rho(g)$ \emph{structures} $g$, i.e., $\rho(g)$ is not a leaf and, letting  $n_1$ and $n_2$ be its two children in some order, we have
    $\vars(g_i) \subseteq
    \vars(T_{n_i} )$ for $i=1,2$. 
    A DNNF is \emph{structured}, written
    \emph{SDNNF}, if there exists a v-tree that structures it.

\item \textbf{Ambiguity level.} A Boolean circuit is \emph{unambiguous}, which is (unfortunately) called for historical reasons \emph{deterministic}, if, intuitively, the inputs to $\lor$-gates are mutually exclusive. Formally,
an~$\lor$-gate~$g$ of~$C$ is \emph{deterministic} if
for every pair~$g_1\neq g_2$ of input gates of~$g$, the Boolean functions
over~$\X$ captured by~$C_{g_1}$ and~$C_{g_2}$ are disjoint; that is, we
have~$\sat(C_{g_1}) \cap \sat(C_{g_2}) = \emptyset$.  We call~$C$
\emph{deterministic} if each~$\lor$-gate is. Note that being deterministic is a semantic condition, not a syntactic one.
    
A stronger notion of determinism is that of determinism in the sense of BDDs,
    which we call \emph{decision} to distinguish it from determinism. An
    $\lor$-gate is \emph{decision} if it has exactly two inputs~$g_0$ and $g_1$
    and there is a variable $X$ such that $g_1$ is an $\land$-gate with $X$ as
    input, and~$g_0$ is an $\land$-gate with a negation gate of $X$ as input.
    The Boolean circuit $C$ is \emph{decision} if all $\lor$-gates are decision. Note
    that being decision is a syntactic condition that can be checked in linear time, and a decision circuit is
    always deterministic.
\end{itemize}

\begin{table}
\caption{Classification of NNF Boolean circuits based on ambiguity level (nondeterministic,
  unambiguous, or deterministic) and variable structuredness (unrestricted,
  free, or ordered).\label{tab-circ}}
\begin{center}
\begin{tabular}{lccc}\toprule
  & \bfseries Unrestricted & \bfseries Decomposable & \bfseries Structured\\\midrule
 \bfseries Arbitrary & NNF & DNNF & SDNNF\\
  {\bfseries Deterministic} (i.e., unambiguous) & d-NNF & d-DNNF &
  d-SDNNF\\
 \bfseries Decision & dec-NNF & dec-DNNF & dec-SDNNF\\\bottomrule
\end{tabular}
\end{center}
\end{table}

The combination of the previous two dimensions again gives rise to 9 different
classes of Boolean circuits, which are shown in Table~\ref{tab-circ}.
A Boolean circuit which is both decomposable and deterministic is
called a \emph{d-DNNF}, and if it is in addition structured then we have a
\emph{d-SDNNF}. A Boolean circuit which is both decomposable and decision is called a
\emph{dec-DNNF}, and \emph{dec-SDNNF} if it is structured.
Note that dec-DNNFs are also called
\emph{decision-DNNFs}~\citep{lagniez2017improved};
they can also be seen as \emph{AND-FBDDs} with decomposable ANDs~\citep{beame2013lower}.

Last, we point out that ambiguity levels are typically only considered in
combination with decomposability, i.e., there is no standard notation for a
deterministic NNF such as \emph{d-NNF} (or \emph{dec-NNF}). This is because, to
the best of our knowledge, there is no interesting task that can be tractably
solved specifically on such Boolean circuits.

\begin{example} \label{ex:class}
The following example is taken from~\citep{arenas2023complexity}. We want to
  classify papers submitted to a conference as rejected (Boolean value $0$) or
  accepted (Boolean value $1$). Papers are described by Boolean variables  
  \feat{fg}, \feat{dtr}, \feat{nf} and \feat{na}, which stand for ``follows
  guidelines", ``deep theoretical result", ``new framework" and ``nice
  applications", respectively.  The Boolean classifier for the papers is given
  by the Boolean circuit in Figure~\ref{fig:ddbc-exa}. The input of this
  Boolean circuit are the variables \feat{fg}, \feat{dtr}, \feat{nf} and \feat{na}, each
  of which can take value either $0$ or $1$, depending on whether the variable  
  is present~($1$) or absent~($0$). The nodes with labels~$\neg$, $\lor$
  or~$\land$ are logic gates, and the associated Boolean value of each one of
  them depends on the logical connective represented by its label and the
  Boolean values of its inputs. The output value of the Boolean circuit is given by the
  top node in the figure.

The Boolean circuit in Figure~\ref{fig:ddbc-exa} is decomposable, because each
  $\land$-gate has two inputs, and the sets of variables of its inputs are
  pairwise disjoint. For instance, in the case of the top node in
  Figure~\ref{fig:ddbc-exa}, the left-hand side input has $\{\feat{fg}\}$ as
  its set of variables, while its right-hand side input has $\{\feat{dtr},
  \feat{nf}, \feat{na}\}$ as its set of variables, which are disjoint. Also,
  this Boolean circuit is deterministic as for every $\lor$-gate two of its
  inputs cannot be given value 1 by the same Boolean assignment for the
  variables. For instance, in the case of the only $\lor$-gate in
  Figure~\ref{fig:ddbc-exa}, if a Boolean assignment for the variables gives
  value 1 to its left-hand side input, then variable \feat{dtr} has to be given
  value 1 and, thus, such an assignment gives value $0$ to the right-hand side
  input of the~$\lor$-gate. In the same way, it can be seen that if a Boolean
  assignment for the variables gives value 1 to the right-hand side input of
  this $\lor$-gate, then it gives value $0$ to its left-hand side input.

  Last note that the Boolean circuit is structured by the vtree shown at the left. Further, the circuit is not decision; however, it could be made decision by adding a $\land$-gate having as inputs the variable gate \feat{dtr} and a constant $1$-gate, and making that gate an input to the $\lor$-gate instead of \feat{dtr}.
\end{example}

\begin{figure}
\begin{subfigure}{0.4\textwidth}
\centering
    \begin{tikzpicture}
        \node[inner sep=0] (v1) at (0, 0) {};
        \node[inner sep=0] (v2) at (1, -1) {};
        \node[inner sep=0] (v3) at (2, -2) {};
        \node (v4) at (-1, -1) {\feat{fg}};
        \node (v5) at (0, -2) {\feat{dtr}};
        \node (v6) at (1, -3) {\feat{nf}};
        \node (v7) at (3, -3) {\feat{na}};
        \draw (v4) -- (v1);
        \draw (v2) -- (v1);
        \draw (v5) -- (v2);
        \draw (v3) -- (v2);
        \draw (v6) -- (v3);
        \draw (v7) -- (v3);
    \end{tikzpicture}
    \caption{Vtree}
\end{subfigure}
\begin{subfigure}{0.58\textwidth}
\centering
\begin{tikzpicture}
  \node[circ, minimum size=7mm, inner  sep=-2] (n1) {\feat{dtr}};
  \node[circ, right=12mm of n1, minimum size=7mm] (n2) {\feat{nf}};
  \node[circ, above=1.3mm of n2, minimum size=7mm] (nneg) {$\neg$}
    edge[arrin] (n1);
  \node[circ, right=12mm of n2, minimum size=7mm] (n3) {\feat{na}};
  \node[circ, right=6mm of nneg, minimum size=7mm] (nbla) {$\land$}
  edge[arrin] (n2)
  edge[arrin] (n3);
  \node[circ, above=10mm of n3, minimum size=7mm] (n4) {$\land$}
  edge[arrin] (nneg)
  edge[arrin] (nbla);
  \node[circ, above=18mm of n2, minimum size=7mm] (n5) {$\lor$}
  edge[arrin] (n4)
  edge[arrin] (n1);
  \node[circ, above=27mm of n1, minimum size=7mm] (n6) {$\land$}
  edge[arrin] (n5);
  \node[circw, left=12mm of n5, minimum size=7mm] (n5a) {};
  \node[circ, left=12mm of n5a, minimum size=7mm, inner sep=-2] (n0) {\feat{fg}}
  edge[arrout] (n6);
\end{tikzpicture}
    \caption{Boolean circuit}
\end{subfigure}
\caption{A deterministic and decomposable Boolean circuit as a classifier and its vtree. \label{fig:ddbc-exa}}
\end{figure}
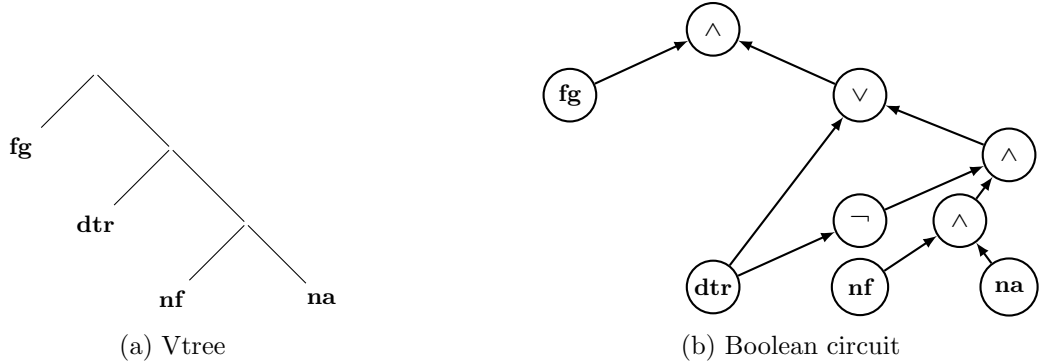

\subsection{Sentential decision diagrams}

\emph{Sentential decision diagrams} (SDDs) are a restricted class of structured
Boolean circuits, satisfying a new ambiguity level called \emph{strong determinism}. An $\lor$-gate $g$ is said to be strongly deterministic if it satisfies the following criteria. Firstly, the inputs of $g$ are all $\land$-gates, written $g_1, ..., g_m$, and they are all structured by some v-tree node $n$, i.e., we have $\rho(g_i) = n$ for all $1 \leq i \leq m$. Further, each $g_i$ for $1 \leq i \leq m$ has precisely two inputs $p_i$ and $s_i$, and, letting $n_1, n_2$ be the children of~$n$ in the order given by the v-tree, we have $\vars(p_i) \subseteq \vars(n_1)$ and $\vars(s_i) \subseteq \vars(n_2)$. Then we call $p_1, ..., p_m$ the \emph{primes} and $s_1, ..., s_m$ the \emph{subs} of $g$, and we require that the primes $p_i, ..., p_m$ are mutually exclusive, namely, we require that $p_i \land p_j$ is unsatisfiable for any $i\neq j$.

A Boolean circuit~$C$ is then \emph{strongly deterministic} if every $\lor$-gate of~$C$ is. This gives rise to the class of
\textit{strongly deterministic SDNNF}~\citep{pipatsrisawat2010lower}, which, as
the name suggests, satisfies a stronger notion of determinism than a d-SDNNF
while being more general than dec-SDNNF. A \textit{sentential decision diagram
(SDD)}~\citep{darwiche2011sdd} is a
strongly deterministic DNNF
which further ensures that, for every $\lor$-gate $g$, letting $p_1, \ldots, p_m$ be its primes, then they are \textit{exhaustive}, formally, $\bigvee_i p_i = \top$. 

\begin{example}
    In Figure \ref{fig:strong-det}, we show an example of a sentential decision
    diagram from \citep{darwiche2011sdd}, in circuit notation. The inputs of
    this Boolean circuit are the variables $A, B, C, D$. It can be seen that structuredness holds with respect to the shown vtree; for example, the $\land$-node immediately below the root $\lor$-node has one child with variables $\{B, A\} \subseteq \{B, A\}$ and one child with variables $\{C\} \subseteq \{D, C\}$. Strong determinism also holds; for example, for the root $\lor$-node, its children are all $\land$-nodes structured by the root vtree node, and the primes correspond to the logical formulae $B \land A$, $B \land \neg A$ and $\neg B$, which are clearly mutually exclusive (and also exhaustive).
\end{example}

\begin{figure}
\begin{subfigure}{0.2\textwidth}
\centering
    \begin{tikzpicture}
        \node[inner sep=0] (v1) {};
        \node[inner sep=0,below=1.5cm of v1, xshift=-1cm] (v2) {};
        \node[inner sep=0,below=1.5cm of v1, xshift=1cm] (v3) {};
        \node[below=1.5cm of v2, xshift=-0.5cm] (v4) {$B$};
        \node[below=1.5cm of v2, xshift=0.5cm] (v5) {$A$};
        \node[below=1.5cm of v3, xshift=-0.5cm] (v6) {$D$};
        \node[below=1.5cm of v3, xshift=0.5cm] (v7) {$C$};
        \draw (v1) -- (v2);
        \draw (v1) -- (v3);
        \draw (v2) -- (v4);
        \draw (v2) -- (v5);
        \draw (v3) -- (v6);
        \draw (v3) -- (v7);
    \end{tikzpicture}
    \caption{Vtree}
\end{subfigure}
\begin{subfigure}{0.78\textwidth}
\begin{center}
\begin{center}
\begin{tikzpicture}
    \node[circ, minimum size=7mm, inner sep=-2] (n1) {$\lor$};
    \node[circ, below=3mm of n1, xshift=-36mm, minimum size=7mm] (n2) {$\land$} edge[arrout] (n1);
    \node[circ, below=3mm of n1, xshift=36mm, minimum size=7mm] (n4) {$\land$} edge[arrout] (n1);
    \node[circ, below=3mm of n1, minimum size=7mm] (n3) {$\land$} edge[arrout] (n1);
    \node[circ, below=3mm of n2, xshift=-8mm] (n5) {$\lor$} edge[arrout] (n2);
    \node[circ, below=3mm of n2, xshift=8mm, inner sep=3] (n6) {$\top$} edge[arrout] (n2);
    \node[circ, below=3mm of n3, xshift=-8mm] (n7) {$\lor$} edge[arrout] (n3);
    \node[circ, below=3mm of n4, xshift=8mm] (n10) {$\lor$} edge[arrout] (n4);
    \node[circ, below=3mm of n5, xshift=-12mm] (n12) {$\land$} edge[arrout] (n5);
    \node[circ, below=3mm of n5, xshift=12mm] (n13) {$\land$} edge[arrout] (n5) edge[arrout] (n7);
    \node[circ, below=3mm of n5, xshift=36mm] (n14) {$\land$} edge[arrout] (n7);
    \node[circ, below=3mm of n10, xshift=-8mm] (n15) {$\land$} edge[arrout] (n10);
    \node[circ, below=3mm of n10, xshift=8mm] (n16) {$\land$} edge[arrout] (n10);

    \node[circ, below=18mm of n12, xshift=-8mm, inner sep=3] (n17) {$B$} edge[arrout] (n12) edge[arrout] (n14);
    \node[circ, below=18mm of n12, xshift=8mm] (n18) {$\neg$} edge[arrout] (n4) edge[arrout] (n13);
    \node[circ, below=3mm of n18] (n19) {$B$} edge[arrout] (n18);
    \node[circ, below=18mm of n12, xshift=40mm] (n21) {$\neg$} edge[arrout] (n14);
    \node[circ, below=3mm of n21] (n22) {$A$} edge[arrout] (n21);
    \node[circ, below=18mm of n12, xshift=24mm, inner sep=3] (n23) {$A$} edge[arrout] (n12);

    \node[circ, below=18mm of n12, xshift=72mm] (n24) {$\neg$} edge[arrout] (n16);
    \node[circ, below=3mm of n24] (n25) {$D$} edge[arrout] (n24);
    \node[circ, below=18mm of n12, xshift=56mm, inner sep=3] (n26) {$D$} edge[arrout] (n15);

    \node[circ, below=18mm of n12, xshift=90mm] (n27) {$C$} edge[arrout] (n3) edge[arrout] (n15);
    \node[circ, below=18mm of n12, xshift=106mm] (n27) {$\false$} edge[arrout] (n13) edge[arrout] (n16);

\end{tikzpicture}
\end{center}
\caption{Sentential decision diagram \label{fig:strong-det}}
\end{center}
\end{subfigure}
\caption{An example of a sentential decision diagram (SDD) and its vtree.}
\end{figure}
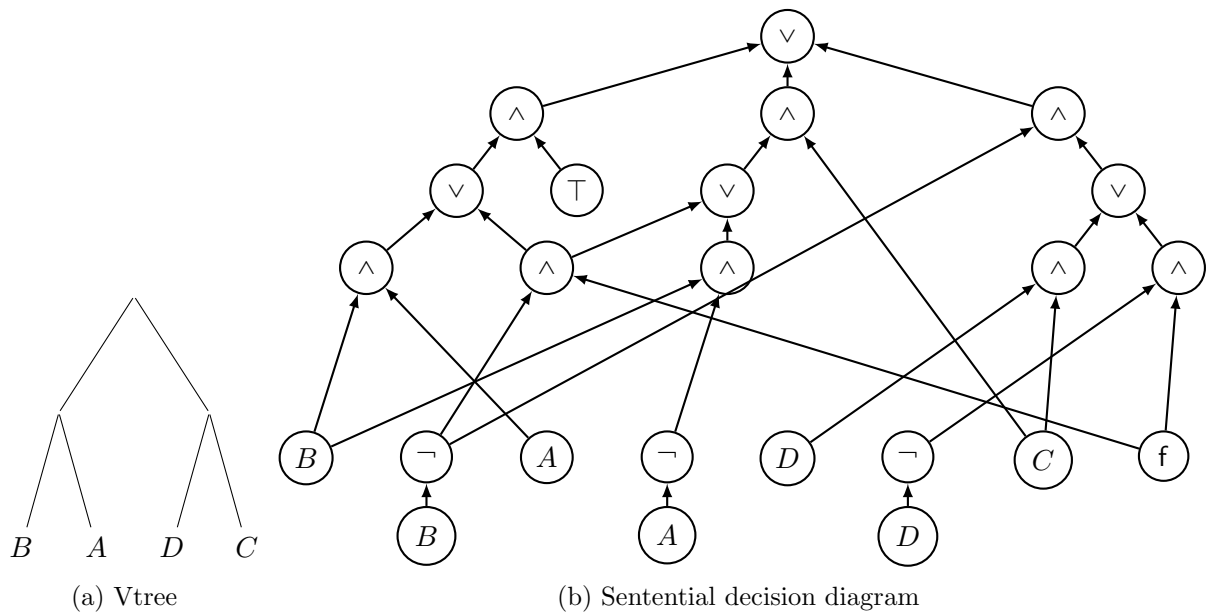

\subsection{Formulas}

As for binary decision diagrams and their restriction to decision trees, one can also
be interested in circuit classes in which we disallow sharing.

When the underlying graph of a Boolean circuit is a tree (noting that we may have
multiple variable gates labeled by the same variable), then we call the result a \emph{Boolean formula}
(not to be confused with a Boolean function). If the Boolean circuit is in NNF,
then the Boolean formula is in NNF. Examples of Boolean formulas in NNF
include Boolean formulas in conjunctive normal form (CNF), 
and Boolean formulas in
disjunctive normal form (DNF).
Note that Boolean formulas may still contain several occurrences
of the same variable or literal; when we further impose that the Boolean circuit
has a tree structure and that there is only one gate for each variable, we
obtain what is called a \emph{read-once formula}~\citep{angluin1993learning}.

\subsection{Smoothness}

Recall that, for binary decision diagrams, we introduced a notion of
\emph{completeness} which intuitively requires binary decision diagrams to test all variables. We now present the analogous requirement on circuits, which is called \emph{smoothness}.

The notion of \emph{smooth} Boolean circuits has been introduced
in~\citep{darwiche2001tractable}; see also~\citep{DBLP:conf/nips/ShihBBA19}.
An $\lor$-gate $g$ of a Boolean circuit $C$ is \emph{smooth} if for every input $g'$ of
$g$ we have $\vars(g)=\vars(g')$, and $C$ is called \emph{smooth} if all its
$\lor$-gates are smooth. Smoothness is the Boolean circuit analogue of
\emph{completeness} for binary decision diagrams. \emph{Smoothing} is the process of
taking a Boolean circuit~$C$ as input, and constructing a Boolean circuit~$C'$ that is
equivalent to $C$ and that is smooth. This can be done naïvely in quadratic
time as follows. Letting $\X$ be the variables, compute in time $O(|C|\cdot |\X|)$
the set $\vars(g)$ for every $g\in \X$. Then, for every $\lor$-gate $g$ and
input gate $g'$ such that $\vars(g') \subsetneq \vars(g)$, replace $g'$ by an
$\land$-gate of $g'$ and of $\bigwedge_{X \in \vars(g)\setminus \vars(g')} (X
\lor \lnot X)$. It is clear that the resulting Boolean circuit is equivalent and
smooth, and that this can be done in $O(|C|\cdot |\X|)$. Furthermore, if $C$ is
in NNF, then so is~$C'$, and if $C$ is
decomposable, then $C'$ also is. We note
that smoothing can be done more efficiently for structured Boolean circuit
classes~\citep{DBLP:conf/nips/ShihBBA19}. 

\section{From Binary Decision Diagrams to Boolean Circuits} \label{sec:diag_circ}
In this section, we connect the notions of binary decision diagrams
(Section~\ref{sec:diagrams}) and Boolean circuits (Section~\ref{sec:circuits}). We do
so by explaining why binary decision diagrams are in fact specific classes of
Boolean circuits, with a
direct translation from binary decision diagrams to Boolean circuits.
Furthermore, we show how the conditions that we have defined on diagrams can be
preserved by this transformation.

\paragraph*{Linear v-trees.}
The translation from binary decision diagrams to Boolean circuits that we will present,
when applied to ordered binary decision diagrams
(\nobdds), will produce structured Boolean circuits whose v-trees have a specific shape.
We define these v-trees, called \emph{right-linear v-trees}, and explain how
to obtain them from the order respected by the \nobdd.

Formally, given an order $<$ on variables $\X$, we define the
\emph{right-linear vtree} $T_<$ obtained from~$<$ in the following way:
\begin{itemize}
  \item if
$\X$ consists of a single variable $x$ then $T_<$ is the singleton tree whose root
is a leaf node labeled by $x$;
\item otherwise, letting $x$ be the smallest element of~$\X$ according to~$<$,
  we let $T_<$ be the tree whose root is an internal node having $x$ has
    left child and having as right child the root of the tree $T_{<'}$ obtained
    from the order $<'$ which is the restriction of~$<$ to $\X \setminus \{x\}$.
\end{itemize}

Note that the choice of right-linear v-trees is arbitrary; similar
constructions could be made with a left-linear v-tree.

\paragraph*{Converting binary decision diagrams to Boolean circuits.} 
We are now ready to state the immediate translation from binary
decision diagrams to
Boolean circuits, and to explain how conditions on the Boolean circuit can be rephrased to
conditions on the binary decision diagram. This is a folklore observation which has previously
appeared, e.g., as Proposition~3.4 of~\citep{ACMS20}:

\begin{proposition}
  Given a \nbdd $\D$, we can convert it in linear time to a Boolean circuit $C$
  that describes the same Boolean function. The translation preserves the
  following:
  \begin{itemize}
  \item Variable structuredness:
  \begin{itemize}
    \item If the input \nbdd is \emph{free},
    then the resulting Boolean circuit is
    \emph{decomposable} (DNNF).
  \item If the input \nbdd is \emph{ordered} with an order $<$, then the
    resulting Boolean circuit is \emph{structured} with the right-linear vtree $T_<$.
  \end{itemize}

  \item Ambiguity level:
  \begin{itemize}
    \item If the input \nbdd is \emph{unambiguous} then the resulting
      Boolean circuit
      is \emph{deterministic}.
    \item If the input \nbdd is \emph{deterministic} then the resulting
      Boolean circuit
      is \emph{decision}.
  \end{itemize}
  
\item Other conditions:
  \begin{itemize}
    \item If the input is a \emph{decision tree}, then the output is a
      \emph{Boolean formula}.
    \item If the input is \emph{complete}, then the result is \emph{smooth}.
  \end{itemize}
  \end{itemize}
\end{proposition}

\begin{proof}
  We use the following general linear-time translation
from binary decision diagrams to Boolean circuits, building a Boolean circuit
  $C$ from the input binary decision diagram $\D$:
\begin{itemize}
  \item We replace true
sinks and false sinks respectively by
constant true and false gates.
    \item We replace
internal nodes $n$ on a variable $X$ with a
gate defined like in the definition of decision
gates above; formally, let $g^1_0,\ldots, g^{k_0}_0$ (resp., $g^1_1,\ldots, g^{k_1}_1$)
be the translations of the nodes 
to which $n$ had edges labeled
with $0$ (resp., with $1$).
Construct a gate $g_0$ (resp., $g_1$) to be an $\lor$-gate of the gates $g^1_0,\ldots, g^{k_0}_0$ (resp., of $g^1_1,\ldots, g^{k_1}_1$).
We then translate
$n$ to a $\lor$-gate whose inputs are an
$\land$-gate conjoining $g_1$ and $X$,
and an $\land$-gate conjoining $g_0$
    and $\neg X$. (If $g_0$ has only one input gate, then we
    replace it by that input, and likewise for~$g_1$.)
    \item Last, the output (root) gate of the
Boolean circuit is an $\lor$-gate taking the disjunction of the translations of all the
    sources of the \nbdd; or, if the \nbdd has only one source, then it is the
    gate that translates this source.
\end{itemize}

  The translation process is illustrated in Figure~\ref{fig:diag2circexa}.
One can check that the translation runs in
linear time and produces a Boolean circuit $C$ with the
same semantics as the \nbdd~$\D$, i.e., $C$ that represents the same
  Boolean function as $\D$. Further, we can check that $C$ has the stated
  properties:
  \begin{itemize}
    \item If the input \nbdd~$\D$ is free, for every $\land$-gate $g$ in~$C$
      created when translating a node $n$ of the \nbdd, then $g$ will
      conjoins a literal for the variable $x$ tested by~$n$ with a gate $g'$,
      and the gates $g''$ having a directed path to~$g'$ are gates $g''$ that are translations
      of nodes $n'$ of~$\D$ to which $n$ has a path (along with intermediary
      gates introduced in the translation), so these gates~$g''$ cannot test $x$
      because~$\D$ is free. Hence, $g$ is decomposable.
    \item If the input \nbdd~$\D$ is ordered by an order~$<$, then any
      $\land$-gate $g$ in~$C$ created when translating a node~$n$ of~$\D$ will
      conjoin a literal for the variable~$x$ tested by~$n$ with a gate that
      depends on variables tested by nodes to which~$n$ has a directed path,
      therefore, $C$ is structured by the right-linear v-tree $T_<$.
    \item If the input \nbdd~$\D$ is unambiguous, then, for every node $n$
      of~$\D$, 
      for every assignment~$\as$ of the variables there is at most one
      accepting run which is compatible with~$\as$ and that goes through $n$. This ensures that,
      considering the $\lor$-gates $g_0$ and $g_1$ created in~$C$ when
      translating~$n$, there cannot be two inputs of such a gate that are made
      true by~$\as$. Further, for every assignment, we know that at most
      one source of~$\D$ is the beginning of a successful run, so the $\lor$-gate
      which is the output gate of~$C$ also does not have two mutually
      satisfiable inputs. (Note that the $\lor$-gates created as the translation
      of~$n$ above are decision gates, so they are always deterministic.)
    \item If the input \nbdd~$\D$ is deterministic, then the gates of the form
      $g_0$ and $g_1$ in the translation had only one input, so they were merged
      with that input; and likewise $\D$ has only one source so we did not
      create an~$\lor$-gate as the output gate of~$C$. Hence, in this case,
      all~$\lor$-gates in~$C$ are those that translate a node of~$\D$, and they
      are decision gates.
    \item If the input $\D$ is a decision tree, then there is no sharing in the
      process described above (creating different copies of variable gates
      labeled by the same variable), so it creates a Boolean formula.
    \item If the input $\D$ is complete, then an easy induction shows that the
      set of variables having a directed path to a gate $g$ of the
      Boolean circuit
      created to translate a node~$n$ of~$\D$ is
      precisely the set of variables tested by the nodes $n'$ of~$\D$ to
      which~$n$ have a directed path, so as~$\D$ is complete we conclude
      that~$C$ is smooth. \qedhere
  \end{itemize}
\end{proof}

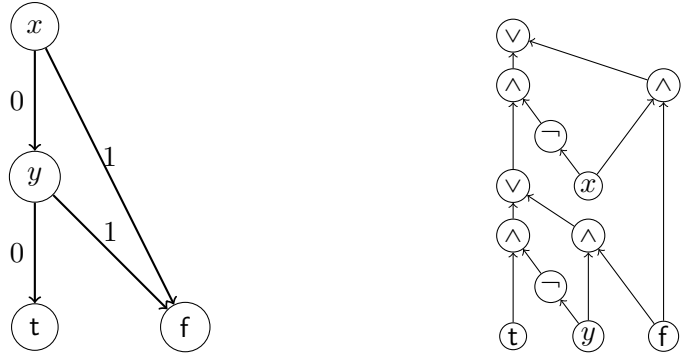
\begin{figure}
  \hfill
  \begin{tikzpicture}[scale=2]
    \node[circle,draw] (a) at (0, 0) {$x$};
    \node[circle,draw] (b) at (0, -1) {$y$};
    \node[circle,draw] (t) at (0, -2) {$\true$};
    \node[circle,draw] (f) at (1, -2) {$\false$};
    \draw[thick,->] (a) -- node[left] {0} (b);
    \draw[thick,->] (b) -- node[left] {0} (t);
    \draw[thick,->] (a) -- node[above] {1} (f);
    \draw[thick,->] (b) -- node[above] {1} (f);
  \end{tikzpicture}
  \hfill
  \begin{tikzpicture}[scale=2]
    \node[circle,draw,inner sep=1pt] (o) at (0, 0) {$\lor$};
    \node[circle,draw,inner sep=1pt] (a1) at (0, -.33) {$\land$};
    \node[circle,draw,inner sep=1pt] (a2) at (1, -.33) {$\land$};
    \node[circle,draw,inner sep=1pt] (a1n) at (.25, -.67) {$\neg$};
    \node[circle,draw,inner sep=1pt] (x) at (.5, -1) {$x$};
    \node[circle,draw,inner sep=1pt] (b) at (0, -1) {$\lor$};
    \node[circle,draw,inner sep=1pt] (b1) at (0, -1.33) {$\land$};
    \node[circle,draw,inner sep=1pt] (b2) at (.5, -1.33) {$\land$};
    \node[circle,draw,inner sep=1pt] (b1n) at (.25, -1.67) {$\neg$};
    \node[circle,draw,inner sep=1pt] (y) at (.5, -2) {$y$};
    \node[circle,draw,inner sep=1pt] (t) at (0, -2) {$\true$};
    \node[circle,draw,inner sep=1pt] (f) at (1, -2) {$\false$};

    \draw[->] (f) -- (a2);
    \draw[->] (f) -- (b2);
    \draw[->] (t) -- (b1);
    \draw[->] (b1n) -- (b1);
    \draw[->] (y) -- (b1n);
    \draw[->] (y) -- (b2);
    \draw[->] (a1n) -- (a1);
    \draw[->] (x) -- (a1n);
    \draw[->] (x) -- (a2);
    \draw[->] (b) -- (a1);
    \draw[->] (a1) -- (o);
    \draw[->] (a2) -- (o);
    \draw[->] (b1) -- (b);
    \draw[->] (b2) -- (b);

  \end{tikzpicture}
  \hfill\null
  \caption{Left: OBDD for the Boolean function $\neg x \land \neg y$. Right:
  equivalent dec-SDNNF.}
  \label{fig:diag2circexa}
\end{figure}

\section{Automata} \label{sec:automata}
This section presents the notions of \emph{word automata} and \emph{tree
automata}. It then explains how to relate these notions to binary
decision diagrams and Boolean circuits, respectively, via the notion of \emph{provenance
circuits} for automata. The use of Boolean circuits for provenance representations is
originally from~\citep{deutch2014circuits}, and its use for automata is
in~\citep{amarilli2015provenance}; but the relation between word automata and
binary decision diagrams has been studied in other contexts,~e.g., \citep{bollig2016minimization}.

The section first defines word automata, then tree automata, and then explains how to compute provenance circuits for these automata. 

\subsection{Automata on words}
\label{subsubsec:ufas-words}

We define standard notions from formal language theory, before defining word automata.

\paragraph*{Alphabets, words, languages.}
An \emph{alphabet} is a finite set~$\Sigma$ of \emph{letters}.  A \emph{word
on~$\Sigma$} is a finite (possibly empty) sequence $\w = w_1,\ldots,w_n$ of
letters from~$\Sigma$; its \emph{length}~$|\w|$ is~$n$.  The set of all words
over~$\Sigma$ is denoted by~$\Sigma^*$.  
A \emph{language} over~$\Sigma$ is a subset of~$\Sigma^*$. 

\paragraph*{Word automata.}
A \emph{non-deterministic finite automaton} (NFA)~$\A = (Q,I,F,\delta)$
over~$\Sigma$ consists of a finite set~$Q$ of \emph{states}, a set $I \subseteq Q$ of \emph{initial
states}, a set~$F \subseteq Q$ of \emph{final states}, and a
\emph{transition relation}~$\delta \subseteq Q \times \Sigma \times Q$.  
We define $|\A|$, the \emph{size} of $A$, to be $|A|  \colonequals |\Sigma| +
|Q| + |\delta|$.
A \emph{partial run of~$\A$ on a word~$\w\in \Sigma^*$} is a sequence of states~$\rho =
q_0, q_1, \cdots, q_{|\w|}$ such that $(q_i, w_{i+1}, q_{i+1}) \in \delta$ for
every~$i \in \{0,\ldots,|\w|-1\}$.
We say that~$\rho$ \emph{starts at~$q_0$}
and~\emph{ends at~$q_{|\w|}$}. 
A \emph{run} of~$\A$ on~$\w$ is a partial run of $\A$ on~$\w$ which starts at an initial state. 
An \emph{accepting run} of~$\A$ on~$\w$ is a run of~$\A$ on~$\w$
which ends in a final state.
For a word~$\w$, we say that
\emph{$\w$ is accepted by~$\A$} if there is an accepting run of~$\A$ on~$\w$.
The \emph{language accepted by~$\A$}, denoted~$L(\A)$, is the set of words
over~$\Sigma$ that are accepted by~$\A$.

The automaton~$\A$ is called
\emph{unambiguous} (UFA) if for every word~$\w \in \Sigma^*$ there exists at
most one accepting run of~$\A$ on~$\w$. Note that this is a semantic condition, which is not straightforward to verify syntactically.
By contrast, we say that the automaton $\A$ is \emph{deterministic} (DFA) if it satisfies the following syntactic condition: there is precisely one initial state in~$I$, and, for every state $q\in Q$ and letter $a \in \Sigma$, there is at most one state $q' \in Q$ such that the transition $(q, a, q')$ is in~$\delta$. In this case, we can equivalently see $\delta$ as a partial function from $Q \times \Sigma$ to~$Q$.
Note that, for any deterministic automaton $\A$ and word~$\w$, then there is at most one run of $\A$ on~$\w$ (accepting or not). Hence,
a deterministic automaton is necessarily unambiguous, but the converse is not true: there are some unambiguous automata that are not deterministic.

We say that a word automaton $\A$ is \emph{trimmed} if, for every state $q\in
Q$, there is a word $\w$ in the language of $\A$ such that $q$ occurs in an
accepting run $\rho$ of $\A$ on~$\w$. Notice that this is equivalent to saying
that, for every state $q\in Q$, there exists a path in the underlying directed
graph of the automaton from some initial state to $q$ (we say that $q$ is
\emph{accessible}), and a path from $q$ to some final state (we say that $q$
is \emph{co-accessible}). Given a word automaton $\A$, we can easily convert
it in linear time to a trimmed automaton which is equivalent (i.e., that
recognizes the same language), and this conversion preserves unambiguity
and determinism.

Note that an unambiguous word automaton~$\A$ which is trimmed satisfies in particular a stronger requirement (*): for any word $\w$, there cannot be two different runs of~$\A$ on~$\w$ that lead to the same state~$q$. Indeed, if this were to happen, then as $q$ is co-accessible we could complete $\w$ to a word on which $\A$ has two accepting runs, contradicting unambiguity. This property~(*) is generally not satisfied if the automaton~$\A$ is not trimmed.

\subsection{Automata on trees}
\label{subsubsec:ufas-trees}

We now define the notions of tree languages and automata over trees.

\paragraph*{Trees.}
Let~$\Sigma$ be an alphabet.  A \emph{$\Sigma$-tree $(T,\lambda)$} is a finite,
rooted, ordered binary tree~$T$ (all internal nodes have exactly two children,
and these are ordered) such that every node~$n$ of~$T$ is labeled by a
letter~$\lambda(n) \in \Sigma$. We say that~$T$ is the \emph{skeleton}
of~$(T,\lambda)$. We denote by~$\mathcal{T}(\Sigma)$ the set of
all~$\Sigma$-trees. A \emph{tree language} over~$\Sigma$ is a (potentially infinite) set of $\Sigma$-trees. The \emph{complement} of a tree language~$L$ is the set of $\Sigma$-trees that are not in~$L$.

\paragraph*{Tree automata.}
We only consider \emph{bottom-up} finite tree automata in this document.
A \emph{non-deterministic (bottom-up) finite tree automaton} (NFTA)
$\A = (Q,F,\Delta,\iota)$ over~$\Sigma$ consists of a finite set of
\emph{states}~$Q$, a set~$F \subseteq Q$ of final states, an
\emph{initialization relation} $\iota \subseteq \Sigma \times Q$, and a \emph{transition
relation} $\Delta \subseteq Q \times Q \times \Sigma \times Q$.
We define $|\A|$, the \emph{size} of $A$, to be $|A|  \colonequals |\Sigma| +
|Q| + |\Delta| + |\iota|$.
Let~$(T,\lambda)$ be a~$\Sigma$-tree.  A \emph{run of~$\A$ on $(T,\lambda)$} is
a~$Q$-tree~$(T,\lambda')$ having the same skeleton $T$ and satisfying the following conditions:
\begin{itemize}
  \item every leaf is labeled by a state given by applying the initialization relation to the label of that leaf: formally, for every leaf~$n$ of~$T$ we have that~$(\lambda(n),\lambda'(n))\in \iota$;
  \item every internal node is labeled by a state given by applying the transition relation to the two states labeling the two children: formally, for every internal node~$n$ of~$T$, letting~$n_1,n_2$ be the (ordered) children of~$n$,
        we have that $(\lambda'(n_1),\lambda'(n_2),\lambda(n),\lambda'(n)) \in \Delta$.
\end{itemize}
We say that run~$(T,\lambda')$ \emph{ends at state~$q$} when~$q = \lambda'(r)$
for~$r$ the root of~$T$.  The run~$(T,\lambda')$ is \emph{accepting} if,
letting~$n$ be the root of~$T$, we have~$\lambda'(n)\in F$. If there exists an
accepting run of~$\A$ on~$(T,\lambda)$ then we say that~$(T,\lambda)$ is
accepted by~$\A$. The \emph{language} $L(\A)$ of $\A$ is the set of $\Sigma$-trees accepted by~$\A$. 

The automaton~$\A$ is called \emph{unambiguous} (UFTA) if
for any~$\Sigma$-tree~$(T,\lambda)$, there exists at most one accepting run
of~$\A$ on~$(T,\lambda)$: again, this is a semantic criterion.
We say that $\A$ is a \emph{deterministic bottom-up finite tree automaton}
(DFTA) if it satisfies the following syntactic criterion:
(1) for every letter $a\in \Sigma$, there is at most one state $q$
such that $(a,q)\in \iota$; and (2) for each pair of states $(q_1, q_2) \in Q$
and letter $a \in \Sigma$, there is at most one state $q$ such that $(q_1, q_2,
a, q) \in \delta$. For a deterministic automaton,
we can equivalently see $\iota$ as a partial function from $\Sigma$ to~$Q$, and see $\delta$ as a partial
function from $Q \times Q \times \Sigma$ to $Q$.
Again, if $\A$ is deterministic, then for any $\Sigma$-tree
$(T,\lambda)$ there is at most one run  of $\A$ on $(T,\lambda)$ (accepting or not).
Hence, again, a deterministic tree automaton is necessarily unambiguous, but the converse does not necessarily hold.

We say that $\A$ is \emph{trimmed} if, for every state $q \in Q$, there is a
$\Sigma$-tree $(T,\lambda)$ in the language accepted by $\A$ and a run $\rho$ of
$\A$ on $(T, \lambda)$ in which $q$ appears. Given a tree automaton, we can
convert it in linear time to an automaton which is equivalent (recognizes the
same language) and is trimmed, and this conversion preserves unambiguity and determinism.

Note that a UFTA $\A$ which is trimmed satisfies again a stronger requirement
(*):
for every
$\Sigma$-tree~$(T,\lambda)$ and state~$q \in Q$ there is at most one run
of~$\A$ on~$(T,\lambda)$ that ends at~$q$.

\subsection{Computing Boolean circuits and binary decision diagrams from automata}

We now explain how to compute Boolean circuits and binary decision diagrams from automata. We first
give the construction for the most general formalisms possible (NFAs and
NFTA), giving nOBDD and structured DNNF circuits respectively. Then, we observe
how these constructions can apply to restricted automata (unambiguous and
deterministic), and which conditions these ensure on the resulting
binary decision diagrams and Boolean circuits.

We assume for now that the alphabet used by automata is $\Sigma = \{0, 1\}$; we will
later explain why other alphabets can also be handled.

\begin{definition}
  Let $\A$ be an NFA over alphabet $\Sigma = \{0, 1\}$, and let $n \in \NN$. The
  \emph{provenance} of~$\A$ on~$n$ is a Boolean function $\phi_{\A, n}$
  over variables
  $\X = \{X_1, \ldots, X_n\}$ defined in the following way: for any assignment $\ba$
  of~$\X$, letting $w_\ba = \as(X_1) \dots \as(X_n)$ be the corresponding word of length $n$ over~ $\Sigma$, we have that $\ba$ satisfies $\phi_{\A, n}$ iff $\A$ accepts $w_\ba$.
  A \emph{provenance circuit} (resp, \emph{provenance decision diagram}) of~$\A$
  on~$n$ is then a Boolean circuit (resp, binary decision diagram) representing
  $\phi_{\A, n}$.

  Likewise, if $\A$ is a NFTA
  over alphabet $\Sigma = \{0, 1\}$ and $T$ is an
  unlabeled full binary tree, the \emph{provenance} of~$\A$ on~$T$ is a Boolean
  function $\phi_{\A, T}$ whose variables $\X$ are the nodes of~$T$ and which is defined in the following way:
  for any assignment $\ba$ of~$\X$, letting $T_\ba$ be the $\Sigma$-tree
  obtained from~$T$ by labeling the nodes according to~$\ba$, we have that $\ba$
  satisfies $\phi_{\A, T}$ iff $\A$ accepts $T_\ba$. 
  A \emph{provenance circuit} (resp., \emph{provenance decision diagram})
  of~$\A$ on~$T$ is a Boolean circuit (resp., binary decision diagram) representing
  $\phi_{\A, T}$.
\end{definition}

We show in this section the following two results on word automata and tree
automata:

\begin{proposition}
  \label{prp:wordprov}
  Let $\A$ be an NFA over alphabet $\Sigma = \{0, 1\}$ and let $n \in \NN$. We
  can build in time $O(|\A| \times n)$ a provenance diagram $C$ of~$\A$ on~$n$
  which is a complete nOBDD with variable order $X_1, \ldots, X_n$.
  Further, if $\A$ is unambiguous then $C$ is an uOBDD, and if $\A$ is
  deterministic then $C$ is an OBDD.
\end{proposition}

\begin{proposition}
  \label{prp:treeprov}
  Let $\A$ be a NFTA over alphabet $\Sigma = \{0, 1\}$ and let $T$ be an
  unlabeled full binary tree. We
  can build in time $O(|\A| \times |T|)$ a provenance circuit $C$ of~$\A$ on~$T$
  which is a smooth structured DNNF (smooth SDNNF), together with its vtree.
  Further, if $\A$ is unambiguous then $C$ is a d-SDNNF.
\end{proposition}

Note that, in Proposition~\ref{prp:treeprov}, there is no discussion of the case
where the input tree automaton is deterministic. We are not aware of a standard
Boolean circuit class corresponding to deterministic automata, though a notion
of \emph{upwards-deterministic Boolean circuits} is introduced for that purpose in~\citep{amarilli2017circuit}.

We briefly comment on the relationship of these results
to~\citep{amarilli2015provenance}. In the latter work, the provenance for
automata is defined on an alphabet consisting of a fixed part together with a
Boolean annotation. For instance, for word automata, the alphabet is $\Sigma
\times \{0, 1\}$, and the provenance is defined on an input word of~$\Sigma^*$,
to describe which of the Boolean annotations of the word are accepted. The
results above, with alphabet $\{0, 1\}$, allow us to recapture this setting.
Indeed, we can modify automata working on a larger alphabet to first read a
binary representation of the letter in~$\Sigma$ followed by the Boolean annotation.
Applying the results above, and fixing the inputs corresponding to letters to
the intended values, gives us the provenance circuit in the sense
of~\citep{amarilli2015provenance}. This uses the fact that the Boolean circuit
and binary decision diagram classes that we consider are closed under the \emph{conditioning} operation where we force a variable to be equal to a specific value.

We first prove Proposition~\ref{prp:wordprov} as a warm-up:

\begin{proof}[Proof of Proposition~\ref{prp:wordprov}]
  We assume that the input automaton $\A$ is \emph{complete} in the sense that,
  for every state $q$ and every letter $b \in \Sigma$, there is at least one
  transition for letter $b$ on state~$q$. We can clearly make $\A$ complete in
  linear time up to adding a sink state.

  We build a binary decision diagram with decision nodes $g_{i,q}$ for each $1 \leq i \leq n$ and
  for each state~$q$: the node $g_{i,q}$ is labeled with the variable $X_i$.
  The initial nodes are $g_{1, q_0}$ for each initial state $q_0 \in I$.
  We also have sinks $g_{n+1,q}$ for each
  state~$q$: the sink $g_{n+1,q}$ is a 0-sink if $q \notin F$, and it is a 1-sink if $q \in F$.

  Now, for every $1 \leq i \leq n$, we add the following edges: for each
  transition from a state $q$ to state $q'$ when reading symbol $b \in \Sigma$,
  we add an edge to $g_{i,q}$ labeled~$b$ to $g_{i+1,q'}$.

  The construction satisfies the time bound. We modify the result in linear time
  to trim it, removing all nodes that do not have a path from a starting node.
  The result is an nOBDD, because
  each decision node has at least one outgoing 0-edge and 1-edge, and the
  variables are ordered in the right way. Note that it is also complete in the
  sense that every run tests all variables.

  We show correctness by finite induction: for every $0 \leq i \leq n$, given
  a word $w$ of length $i$, the reachable nodes $g_{i+1,q}$ by the partial
  valuation corresponding to~$w$ are those for the states $q$ that we can reach
  when reading~$w$. The base case is immediate: when reading the empty word, we
  can get precisely to the initial state. Now, for the induction step, if the
  states that we can reach when reading a word of length $i$ are correct, then
  when reading an extra letter $b \in \Sigma$ we can go precisely to the states
  having a transition labeled $b$ from the previously reachable states, so the
  result is correct. The finite induction applied to $i=n$ confirms that we can
  reach a final state (i.e., the word is accepted) iff we can reach a 1-sink.

  Now, observe that if the automaton is unambiguous, then for any valuation
  there is at most one way to reach the 1-sink, because there is at most one
  accepting run of the automaton on the corresponding word. Further, if the
  automaton is deterministic, then for any valuation there is at most one
  consistent path, because each decision node has at most one outgoing 0-edge
  and 1-edge.
\end{proof}

We next prove Proposition~\ref{prp:treeprov}. To do this, we first need to do a
small adjustment: v-trees are defined in Section~\ref{sec:circuits}
so that the variables are at the leaves,
but tree automata read labels on all tree nodes, including internal nodes.

\begin{definition}
  Let $T$ be a finite rooted ordered binary tree with nodes $N$. Its
  \emph{leaf-push}
  is the finite rooted ordered binary tree $T'$ obtained by applying bottom-up the following
  transformation: the transformation of a leaf node is this leaf node, and we replace
  each internal node $n$ with children $n_1$ and $n_2$ by an internal node $n'$
  having as children one leaf $n$ and an internal node $n''$ having as children
  the transformations of $n_1$ and $n_2$.

  Note that the leaves of $T'$ are precisely in bijection with the nodes of~$T$.
\end{definition}

\begin{proof}[Proof of Proposition~\ref{prp:treeprov}]
  We do not assume this time that the input tree automaton $\A$ is complete, but
  must instead assume that it is \emph{trimmed}, so as to satisfy property (*)
  above.

  We build a Boolean circuit bottom-up, with $\lor$-gates $g_{n,q}$ for each tree node~$n$ of the
  input tree~$T$.
  The output gate will be an $\lor$-gate $g$ doing the disjunction of all the
  $g_{r,q}$ for $q$ final and for $r$ the root of~$T$.

  Now, for every leaf $n$ of the input tree~$T$, we add
  a variable gate for~$n$ as an input of the gate $g_{n,q}$ for each $q \in \iota(1)$, and
  we add a gate for the negation of~$n$ as an input of the gate $g_{n,q}$ for
  each $q \in \iota(0)$.

  For every internal node $n$ of the input tree $T$ with children $n_1$ and
  $n_2$, for every pair of states $q_1$ and $q_2$, for every $b \in \{0, 1\}$,
  for every state $q$ having a transition from $q_1, q_2, b$, we add as input to
  $g_{n,q}$ a $\land$-gate having as first input either $n$ or $\neg n$
  depending on whether $b=1$ or $b=0$, and as second input a $\land$-gate having
  as first input $g_{n_1,q_1}$ and as second input $g_{n_2,q_2}$.

  We them remove from the obtained Boolean circuit all gates that have no directed path
  to the output gate.

  The construction satisfies the time bound.  Note that
  the circuit is clearly in NNF. We claim that it is decomposable, and that it
  is structured by a v-tree which is the leaf-push of the tree~$T$. Indeed, we
  can easily show by bottom-up induction that the domain of each gate $g_{n,q}$
  is a subset of the nodes of the subtree of~$T$ rooted at~$n$ (including~$n$),
  so that the $\land$-gates are indeed structured according to the leaf-push.
  We additionally point out that the circuit is smooth: this relies on the fact that
  we removed gates having no directed path to the output gate.

  We show correctness again by finite induction: for every tree node~$n$ of~$T$, given
  a labeling of the subtree~$T_n$ of~$T$ rooted at~$n$,
  the gates $g_{n,q}$ satisfied by the corresponding partial assignment are
  precisely those for the states $q$ that we can reach
  when reading the labeling of~$T_n$ in~$\A$.
  The base case corresponds to leaves, where indeed the gates $g_{n,q}$ that are
  satisfied follow the initial function $\iota$ by construction.

  Now, for the induction step, let us consider an internal node~$n$ of~$T$ with
  children $n_1$ and $n_2$. By induction hypothesis, we assume that the
  states that we can reach when reading a labeling $\lambda_1$ of the tree
  $T_{n_1}$ and when reading a labeling $\lambda_2$ of the tree $T_{n_2}$ are the ones given by the invariant.
  Let us show that the invariant is satisfied when reading a labeling $\lambda$ of~$T_n$.
  Note that $\lambda$ is defined by a labeling
  $\lambda_1$ on~$T_{n_1}$, a labeling $\lambda_2$ on $T_{n_2}$, and a label $b \in
  \Sigma$ on~$n$. The gates $g_{n,q}$ that are satisfied are precisely those for
  which there is a transition from $q_1$ and $q_2$ and $b$ to~$q$ and for which
  $g_{n_1,q_1}$ and $g_{n_2,q_2}$ are satisfied: this allows us to conclude from
  the induction hypothesis. Hence, using the induction result on
  the root $r$ of~$T$ (where $T_r = T$), we conclude that the output gate of the
  Boolean circuit is true on an assignment $\as$ iff there is a final state of~$\A$ that we can
  reach on the labeling of~$T$ according to~$\as$.

  Now, observe that if the automaton is unambiguous, then each $\lor$-gate is
  deterministic. Indeed, if the output gate had two inputs that are mutually
  satisfiable, then it witnesses the existence of a labeling of the tree~$T$
  on which $\A$ reaches two different final states, contradicting the
  unambiguity of~$\A$. Likewise, if a gate $g_{n,q}$ has two inputs that are
  mutually satisfiable, as they are satisfied by the same valuation the literal
  on~$n$ must be the same so it must be the case that we can simultaneously
  satisfy $g_{n_1, q_1}$ and $g_{n_2, q_2}$, and $g_{n_1,q_1'}$ and
  $g_{n_2,q_2'}$, for two 2-tuples $(q_1,q_2) \neq (q_1',q_2')$. This witnesses
  that, on this labeling of~$T_n$, the automaton has two distinct runs leading
  to state~$q$ at the root. This is a contradiction of property (*).
\end{proof}

\section{Conclusion and Extensions} \label{sec:conclusion} 
We have introduced in this document the notions of binary decision diagrams,
Boolean circuits, and automata. We have explained in which sense binary decision diagrams can be seen as a special case of Boolean circuits. We have also
explained how automata can be translated to structured Boolean circuits (for
tree automata) or ordered binary decision diagrams (for word automata).

We close the document by reviewing topics which are not presently covered by the document, but could be covered in further versions of the document or by follow-up works:

\begin{itemize}
    \item The connections between \emph{SDDs and automata}, between
      \emph{strongly deterministic Boolean circuits and automata}, or the
      question of which Boolean circuits can be associated to \emph{deterministic automata} (a related notion is \emph{upwards-determinism} in~\citep{amarilli2017circuit}).
    \item The notion of \emph{$k$-unambiguity} for automata (i.e., having \emph{at most $k$} accepting runs), and the classes of
      Boolean circuits and binary decision diagrams to which this corresponds.
    \item The notion of \emph{width} for structured
      Boolean circuits~\citep{capelli2019tractable} and for binary
      decision diagrams, and its connection to the number of states of automata.
    \item The class of \emph{circuits having a compatible order}~\citep{amarilli2017circuit}, which is intermediate between stucturedness and decomposability (i.e., it restricts the possible conjunctions but in a weaker way than requiring a fixed v-tree).
\item The question of the \emph{complexity of various problems} on the various
  classes of Boolean circuits and binary decision diagrams, in the spirit of the
    knowledge compilation map~\citep{darwiche2002knowledge}; or their
    \emph{closure under operations} (e.g., given two Boolean circuits in some formalism, can we tractably compute their disjunction, their conjunction, etc., in the same formalism?); or the \emph{complexity of translating from one formalism to another} (some results of this kind are surveyed in~\citep{ACMS20})
\item The complexity, given an input Boolean circuit or binary decision diagram, of \emph{testing which conditions it satisfies}: this is generally tractable for syntactic criteria (e.g., decomposability, decision), but can be intractable for semantic criteria (e.g., determinism).
\item The extension from Boolean circuits to different circuit types. These include in particular \emph{multivalued circuits}, which are defined over larger sets than the Boolean values; such circuits can be useful in a database context because they relate to the notion of \emph{factorized databases}~\citep{olteanu2015size}.
We can also interpret circuits more generally over semirings (of which the Boolean semiring is a special case), which has been used in the setting of \emph{provenance circuits for semiring provenance}~\citep{deutch2014circuits}, and which has been exploited for tractable model counting \citep{kimmig2017algebraic}. Last, we can consider \emph{arithmetic circuits}~\citep{shpilka2010arithmetic}, which are circuits
computing a polynomial over a given field~$\mathbb{F}$, with internal nodes corresponding to multiplication or addition. The analogue of decomposability is then to require
    that the polynomial computed is \emph{syntactically multilinear}, i.e., that in the expanded form of the polynomial, in every monomial, each variable has exponent either zero or one.
Further, the analogue of determinism is then that every monomial in the expanded form of the polynomial has coefficient either zero or one.
\item The connections to \emph{probabilistic circuits}~\citep{ProbCirc20}, which are circuit classes that define probability distributions rather than Boolean functions. The field of probabilistic circuits also defines circuit properties analogous to those discussed in this document, and studies the complexity of various \emph{probabilistic queries} on different classes of probabilistic circuits. We could also study the connections to various tractable probabilistic models such as \emph{sum-product
networks}~\citep{poon2011sum}, \emph{arithmetic circuits}~\citep{darwiche2003differential}, \emph{cutset network}~\citep{rahman2014cutset}, \emph{and-or graphs}~\citep{dechter2007and}, \emph{probabilistic sentential decision diagrams}~\citep{kisa2014probabilistic}, and more, which can be understood via translations to various classes of probabilistic circuits.
\item The connections to other formalisms to represent languages over words,
  e.g., \emph{context-free grammars}, for which Boolean circuit representations are implicit in~\citep{amarilli2022efficient}, and which have recently been related to factorized databases in~\citep{kimelfeld2023unifying}.
\item The use of \emph{depth reduction} techniques, to rewrite circuits and
  binary decision diagrams to circuits of lower depth~\citep{valiant1983fast}
\item As a converse to provenance circuits, the \emph{conversion of circuits into automata}, although this is not obvious to define because, unlike automata, circuits generally do not have a ``uniform'' behavior.
\end{itemize}

\bibliographystyle{plainnat}
\bibliography{references}

\end{document}